\pdfoutput=1

\documentclass[10pt,a4paper]{article}

\usepackage{amsmath,amsgen,latexsym}
\usepackage{amstext,amssymb,amsfonts}
\usepackage{theorem}
\usepackage{graphicx}

 \newcommand{\bs}{\bigskip}
 \newcommand{\ms}{\medskip}
 \newcommand{\n}{\noindent}
 \newcommand{\s}{\smallskip}
 \newcommand{\hs}[1]{\hspace*{ #1 mm}}
 \newcommand{\vs}[1]{\vspace*{ #1 mm}}

 \newcommand{\setempty}{\varnothing}
 
 \newcommand{\nat}{\mathbb{N}}
 
 \newcommand{\integer}{\mathbb{Z}}
 \newcommand{\rational}{\mathbb{Q}}

 \newcommand{\co}{\mathrm{co}\mbox{-}}

 \newcommand{\etal}{\textrm{et al.}\hspace*{2mm}}

 \newcommand{\BB}{{\cal B}}
 \newcommand{\CC}{{\cal C}}
 
 \newcommand{\DD}{{\cal D}}

 \newcommand{\SSS}{{\cal S}}

 \newcommand{\VV}{{\cal V}}

 \newcommand{\dl}{\mathrm{L}}
 \newcommand{\nl}{\mathrm{NL}}
 \newcommand{\p}{\mathrm{P}}
 \newcommand{\np}{\mathrm{NP}}

 \newcommand{\fl}{\mathrm{FL}}

\theoremstyle{plain}
\theoremheaderfont{\bfseries}
 \newtheorem{theorem}{Theorem}[section]
 \newtheorem{lemma}[theorem]{Lemma}
\newtheorem{proposition}[theorem]{{\bf Proposition}}
 \newtheorem{corollary}[theorem]{Corollary}

 {\theorembodyfont{\rmfamily}
  \newtheorem{definition}[theorem]{Definition}}
 {\theorembodyfont{\rmfamily} \newtheorem{example}[theorem]{Example}}
 {\theorembodyfont{\rmfamily} }

 \newtheorem{claim}{Claim}
 \newtheorem{yexample}[theorem]{Example}

 \newenvironment{proofof}[1]{%\vspace*{2mm}
         \par \noindent
         {\bf Proof of #1.\hs{2}}}{\hfill$\Box$ \vspace*{3mm}}

 \newenvironment{proof}{\par \noindent
            {\bf Proof. \hs{2}}}{\hfill$\Box$ \vspace*{3mm}}

 \newenvironment{yproof}{\par \noindent
            {\bf Proof. \hs{2}}}{\hfill$\Box$ \vspace*{3mm}}

\newcommand{\ignore}[1]{}

\newcommand{\psublin}{\mathrm{PsubLIN}}

\newcommand{\dstcon}{\mathrm{DSTCON}}

\newcommand{\Lreduces}{\leq^{\mathrm{L}}_{m}}

\newcommand{\Lequiv}{\equiv^{\mathrm{L}}_{m}}

\newcommand{\snl}{\mathrm{SNL}}
\newcommand{\snp}{\mathrm{SNP}}
\newcommand{\para}{\mathrm{para}\mbox{-}}
\newcommand{\boldvec}[1]{\mbox{\boldmath $ #1 $}}

\newcommand{\twosat}{\mathrm{2SAT}}

\newcommand{\logdcfl}{\mathrm{LOGDCFL}}
\newcommand{\logcfl}{\mathrm{LOGCFL}}

\begin{document}
%%%
%%%
\pagestyle{plain}
%\pagenumbering{arabic}
\setcounter{page}{1}
\setcounter{footnote}{0}

\begin{center}
{\Large {\bf Logical Expressibility of Syntactic NL for Complementarity, Monotonicity,  and Maximization}}\footnote{An extended abstract \cite{Yam24} appeared under a slightly different title in the Proceedings of the 30th International Workshop on Logic, Language, Information, and Computation (WoLLIC 2024), Bern, Switzerland, June 10--13, 2024, Lecture Notes in Computer Science,  vol. 14672, pp.  261--277, Springer, 2024.} \bs\ms\\

{\sc Tomoyuki Yamakami}\footnote{Present Affiliation: Faculty of Engineering, University of Fukui, 3-9-1 Bunkyo, Fukui 910-8507, Japan}
\bs\\
\end{center}

%%%%%%%%%%%%%%%%%%

\sloppy
\begin{abstract}
Syntactic NL or succinctly SNL was first introduced in 2017, analogously to SNP, as a ``syntactically''-defined natural subclass of NL (nondeterministic logarithmic-space complexity class) using a restricted form of logical sentences, starting with second-order  ``functional'' existential quantifiers followed by first-order universal quantifiers, in close connection to the so-called linear space hypothesis.
We further explore various properties of this complexity class SNL to achieve the better understandings of logical expressibility in NL.
For instance, SNL does not enjoy the dichotomy theorem unless L$=$NL.
To express the ``complementary''
problems of SNL problems logically, we introduce $\mu\snl$, which is an extension of SNL by allowing the use of $\mu$-terms.
As natural variants of $\snl$, we further study the computational complexity of monotone and optimization versions of SNL, respectively called MonoSNL and MAXSNL.
We further consider maximization problems that are logarithmic-space approximable with only constant approximation ratios.
We then introduce a natural subclass of MAXSNL, called MAX$\tau$SNL, which enjoys such limited approximability.

\s
\n{keywords:} second-order logic, NL, SNL, monotone SNL, optimization problem,  MAX-CUT, MAX-UK
\end{abstract}

%%%%%%%%%%%%%%%%%
%%%%%%%%%%%%%%%%%
\sloppy
\section{Background and Major Contributions}

%%%%
\subsection{Motivational Discussion on Syntactic NL}\label{sec:introduction}

Since its importance was first recognized in the 1970s, the \emph{nondeterministic polynomial-time complexity class} $\np$ has been a centerfold of intensive research in the field of computer science. The ``complexity'' of each $\np$ problem has measured mostly in terms of the algorithmic behaviors of its underlying nondeterministic Turing machine (NTM) that solves it in polynomial time.
From a completely distinct perspective, another significant method in measuring the complexity of $\np$ problems can be given by the logical expressibility of how to describe (or express) a given problem using only logical symbols (i.e., variables, connectives, quantifiers, etc.).
In the late 1990s, Papadimitriou and Yannakakis \cite{PY91} and Feder and Vardi \cite{FV93,FV99} studied a logically-expressed subclass of $\np$, known now as $\mathrm{SNP}$, to capture a certain aspect of nondeterministic polynomial-time computation
in terms of second-order logical sentences starting with a second-order existential quantifier followed by a first-order universal quantifier
(with no use of the first-order existential quantifiers).
As Impagliazzo and Paturi \cite{IP01} demonstrated, the satisfiability problem whose inputs are Boolean formulas of $k$-conjunctive normal form ($k$CNF), $k\mathrm{SAT}$, is complete for $\mathrm{SNP}$ under so-called SERF reductions.
This complexity class $\mathrm{SNP}$ turns out to play an important role in promoting the better understanding of the syntactic expressibility of capturing nondeterministic polynomial-time computing.

\emph{Nondeterministic logarithmic-space (or log-space) computation} is also an important resource-bounded  computation in theory and also in practice.
Analogously to $\np$, such log-space computation formulates the \emph{nondeterministic log-space complexity class} $\nl$.
Typical $\nl$ decision problems include the 2CNF formula satisfiability problem ($\twosat$) and the directed $s$-$t$ connectivity problem ($\dstcon$). Interestingly, numerous properties that have been unknown for $\np$ are already settled for $\nl$ due to the log-space restriction of work tapes of NTMs.
For instance, $\nl$ is closed under complementation \cite{Imm88,Sze88} whereas $\np$ is believed by many researchers not to be closed under the same set operation.
In due course of a study on the complexity of ``parameterized'' decision problems, analogously to SNP, a ``syntactically''-defined  natural subclass of $\nl$ dubbed as \emph{Syntactic NL} (or succinctly, $\emph{SNL}$) and its variant $\snl_{\omega}$ were introduced in \cite{Yam17a} based on restricted forms of second-order sentences starting with second-order ``functional'' existential quantifiers (for their detailed definitions, refer to Section \ref{sec:syntactic-NL}). These logic-based complexity classes have played an important role in the field of parameterized problems with size parameters \cite{Yam17a}.
The logical expressibility of the ``parameterized'' version of SNL (resp.,  SNL$_{(\omega)}$), denoted $\para\snl$ (resp., $\para\snl_{\omega}$) for clarity, was discussed in \cite{Yam17a} within the theory of sub-linear space computation.
The complexity class $\para\snl$ naturally contains a parameterized version of $\nl$-complete problem, known as the directed $s$-$t$ connectivity problem of degree at most $3$ ($3\dstcon$), and $\para\snl_{\omega}$ contains a parameterized version of its variant, called
$\mathrm{exact3DSTCON}$, whose input graphs are restricted to vertices of degree exactly $3$ \cite{Yam17a}.
Moreover, $\para\snl_{\omega}$ is closely related to
a practical working hypothesis, known as the \emph{linear space hypothesis}\footnote{The \emph{linear space hypothesis} (LSH) states the existence of a parameterized (decision) problem that is not solvable in polynomial time using $O(n^{\varepsilon})$ space for any constant $\varepsilon\in[0,1)$. See Section \ref{sec:numbers-machines} for its precise definition.} (LSH), which was also introduced in \cite{Yam17a} and further developed in, e.g., \cite{Yam17b,Yam18a,Yam19,Yam22a,Yam22b}.
This LSH is regarded as a log-space analogue of the exponential time hypothesis
(ETH) and the strong exponential time hypothesis (SETH) of \cite{IP01,IPZ01}. The importance of LSH partly comes from the fact that, Under LSH, we can derive the long-awaiting separations: $\dl\neq\nl$, $\logdcfl\neq \logcfl$, and $\mathrm{SC}\neq\mathrm{NSC}$ \cite{Yam17a}.

Up to now, little is known for the properties of $\snl$ and $\snl_{\omega}$. The power of logical expressibility in the log-space setting has been vastly unexplored.
The primary purpose of this work is therefore to explore their fundamental properties,
in straight comparison with $\snp$.

In the past literature, natural variants of SNP have been studied to promote our basic understandings of the logical expressibility.
Papadimitriou and Yannakakis \cite{PY91} investigated in 1991 an optimization version of $\mathrm{SNP}$, called $\mathrm{MAXSNP}$, in a discussion of the development of fast approximation algorithms. Notably, they showed that MAXSNP is contained in APX; namely, all optimization problems in $\mathrm{MAXSNP}$ can be approximated in polynomial time within certain fixed approximation ratios.
They also demonstrated that many of the typical $\np$ optimization problems, including $\mathrm{MAX}\mbox{-}\mathrm{2SAT}$ and $\mathrm{MAX}\mbox{-}\mathrm{CUT}$, are in fact complete for $\mathrm{MAXSNP}$ under polynomial-time linear reductions (later, Lemma \ref{MAX-CUT-complete} shows that they are complete even under log-space AP-reductions). Other natural MAXSNP-complete problems were discussed in, e.g., \cite{Kan91,Kan92}. Lately, Bringman, Cassis, Fisher, and K\"{u}nnmann \cite{BCFK21} studied a subclass of MAXSNP, called MAXSP.
By taking a similar approach, it is possible to introduce an optimization version of $\snl$, which we intend to call $\mathrm{MAXSNL}$ (Definition \ref{def-MAXSNL}). What similarities and differences lay between $\mathrm{MAXSNL}$ and $\mathrm{MAXSNP}$? Which optimization problems in MAXSNL are approximately solvable using only log space with constant approximation ratios? Those approximable problems form the complexity class APXL \cite{Tan07,Yam13,Yam13b}.

As another variant of $\mathrm{SNP}$, Feder and Vardi \cite{FV93,FV99} studied in the late 1990s natural subclasses of $\mathrm{SNP}$ in hopes of proving  the so-called \emph{dichotomy theorem}, which asserts that all problems in a target complexity class are either in $\p$ or $\np$-complete. In particular, they considered
three restricted subclasses of SNP, namely, monotone SNP, monotone monadic SNP with disequalities, and monadic SNP. Monotone monadic SNP (dubbed as MMSNP), for example, is shown to be polynomially equivalent to constraint satisfaction problems (CSPs).
Notice that the class of all CSPs on the two element domain is proven by Schaefer \cite{Sch78} to enjoy the dichotomy theorem. A characterization of MMSBP under a natural restriction was also discussed in \cite{BCF12}.
A similar approach can be taken to introduce the monotone $\snl$, denoted by MonoSNL (Definition \ref{def-monoSNL}), and its binary variation, called MonoBSNL (Definition \ref{def-BSNL}).
What fundamental properties does MonoSNL own in comparison with $\mathrm{MMSNP}$? Do all CSPs restricted to MonoSNL enjoy the dichotomy theorem?

%%%%%
\subsection{Major Contributions and the Organization of This Work}

%%%%%%
%%%%%%

\begin{figure}[t]
\centering
\includegraphics*[height=5.9cm]{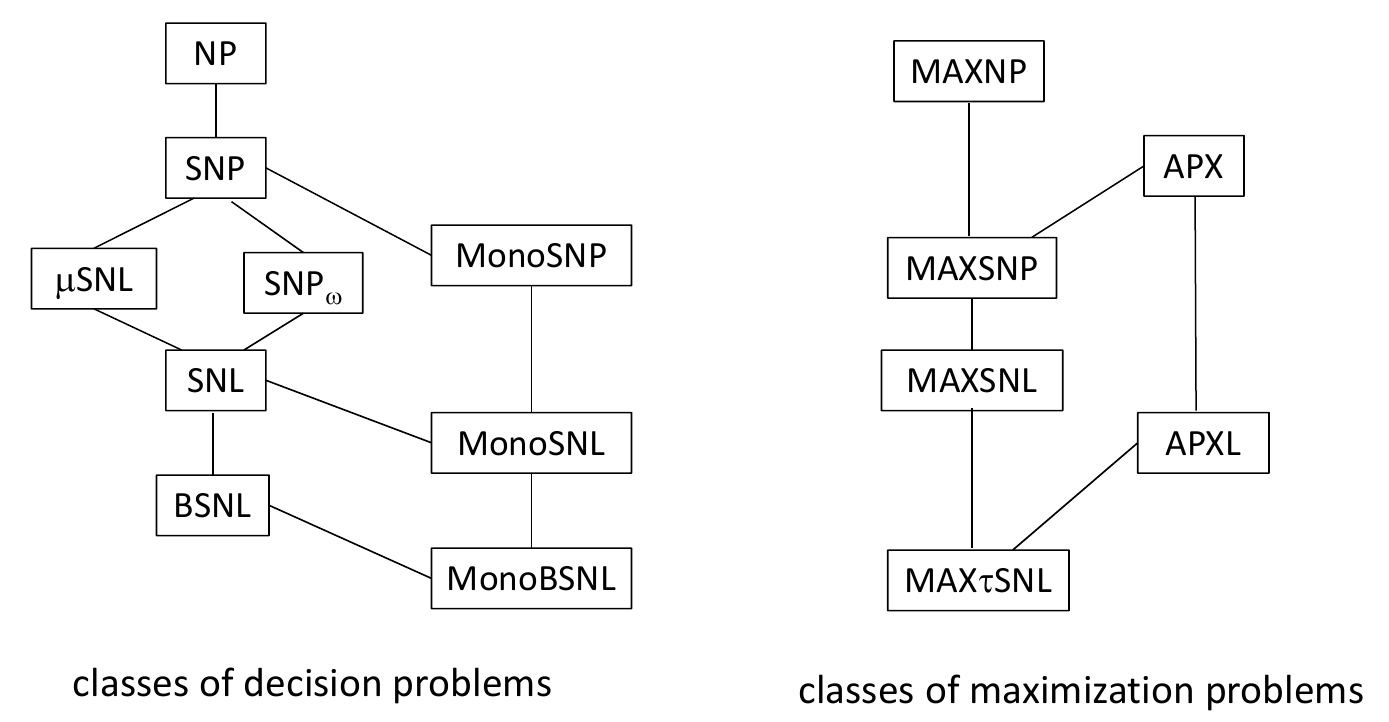}
\caption{Inclusion relationships among complexity classes discussed in this work.}\label{fig:SNL-hierarchy}
\end{figure}

%%%%%
%%%%%

We briefly describe three major contributions of this work on the complexity class $\snl$ and its natural variants defined later in this work. Section \ref{sec:syntactic-NL} will present the brief explanation of $\snl$ and its underlying notions, including vocabulary and relational and domain structures. Figure \ref{fig:SNL-hierarchy} illustrates inclusion relationships among the complexity classes discussed in this work.

The first major result presented in Section \ref{sec:structure} is concerning the structure of decision problems (or equivalently, languages) in $\snl$.
To measure the relative complexity of decision problems in $\snl$, we use \emph{logarithmic-space many-one reductions} (or $\dl$-m-reductions, for short) and
show in Section \ref{sec:basic-closure} that every decision problem in $\nl$ has an $\dl$-m-equivalent problem in $\snl$ (Proposition \ref{equivalent-MonoSNL}), where two problems are said to be \emph{L-m-equivalent} if one problem is reducible to another and vice visa under $\dl$-m-reductions.
This does not seem to imply that SNL is also closed under complementation in spite of $\nl=\co\nl$.
We will show that $\mathrm{2COLOR}$  (2-colorability problem) and its ``complementary'' problem (dubbed as NBG) belong to $\snl$ (Proposition \ref{NBG-mu}).
In contrast, the ``complementary'' problem of $\dstcon$, called $\mathrm{DSTNCON}$, is proven to be contained in a natural superclass of $\snl$ (Theorem  \ref{DSTNCON-mu}), which is called $\mu\snl$. This is proven in Section \ref{sec:mu-operator} by syntactically implementing a well-known technique of \emph{inductive counting} \cite{Imm88,Sze88}. Unfortunately, it remains unknown that $\mathrm{SNL}= \mu\mathrm{SNL}$.

Secondly, we will study in Section \ref{sec:monotone} a relationship between the monotone variant of SNL, called  $\mathrm{MonoSNL}$, and a dichotomy theorem.
A \emph{dichotomy theorem} classifies all languages in question to  either ones that are easy to solve or ones that are extremely difficult to solve. Such a dichotomy theorem is known for several restricted subclasses of $\mathrm{SNP}$.
In comparison, we will show that $\snl$ does not enjoy the dichotomy theorem unless $\dl=\nl$ (Corollary \ref{dichotomy-L-NL}).
We further require second-order functional variables to behave as functions mapping natural numbers to $\{0,1\}$.
We write $\mathrm{BSNL}$ when all of its underlying $\snl$ sentences satisfy this additional requirement.
$\mathrm{MonoSNL}$ and $\mathrm{MonoBSNL}$ (which are respectively monotone versions of $\mathrm{SNL}$ and $\mathrm{BSNL}$) are closely related to restricted forms of constraint satisfaction problems, which we respectively call $\mathrm{CSP}_2(\Gamma)$ and $\mathrm{BCSP}_2(\Gamma)$ for an arbitrary set $\Gamma$ of constraint functions.
Although it is unknown that $\mathrm{MonoBSNL}$ enjoys the dichotomy theorem, we will note in Section \ref{sec:monotone} that $BCSP_2(\Gamma)$ is either in $\dl$ or $\nl$-complete  for any set $\Gamma$  (Corollary \ref{BCSP-dichotomy}).

Thirdly, analogously to MAXSNP, we will consider the class $\mathrm{MAXSNL}$ of optimization problems in Section \ref{sec:maximal}. Instead of polynomial-time linear reductions of \cite{PY91}, we use logarithmic-space AP-reductions of  \cite{Yam13b}. The maximization problem $\mathrm{MAX}\mbox{-}\mathrm{CUT}$ is complete for $\mathrm{MAXSNL}$ under these reductions (Lemma \ref{MAX-CUT-complete}).
Regarding log-space approximation schemes of \cite{Yam13,Yam13b}, we will  construct them for a maximization version of the unary 0-1 knapsack problem (UK), called $\mathrm{MAX}\mbox{-}\mathrm{UK}$ (Proposition \ref{MAXUK-approximation}), and for a subclass of $\mathrm{MAXSNL}$, called  $\mathrm{MAX}\tau\mathrm{SNL}$ (Theorem \ref{MAX-tau-SNL}).

In Section \ref{sec:open-problem}, we will provide a short list of open problems for future research on $\snl$ and its relevant notions.

%%%%%%%
%%%%%%%
\vs{-2}
\section{Basic Notions and Notation}\label{sec:notions}

We briefly discuss basic notions and notation used in the rest of this work.

%%%
\subsection{Numbers, Machines, and Reducibility}\label{sec:numbers-machines}

Concerning numbers, we use three special notations $\nat$, $\integer$, and $\rational$, which respectively consist of all \emph{natural numbers} (including $0$), of all \emph{integers}, and all \emph{rational numbers}.  Moreover, we set $\nat^{+}=\nat-\{0\}$ and $\rational^{>1} = \{r\in\rational\mid r>1\}$. For two numbers $m,n\in\integer$ with $m\leq n$, the notation $[m,n]_{\integer}$ denotes an \emph{integer interval} $\{m,m+1,m+2,\ldots,n\}$. Given a number
$n\in\nat^{+}$, $[n]$ is a shorthand for $[1,n]_{\integer}$.
For a finite set $S$, $|S|$ denotes the \emph{cardinality} of $S$; that is, the total number of elements in $S$. As customary, we freely identify decision problems with their associated languages.

We assume the reader's familiarity with \emph{deterministic Turing machines} (or DTMs, for short) with random-access mechanism, each of which is equipped with a read-only input tape, multiple work tapes, and an index tape used to specify an address of the input tape for an instant access to a target input tape cell.
For any decision problem (which is freely identified with its corresponding language) $L$, a DTM $M$ is said to \emph{solve} $L$ if, for any instance $x$ in $L$, $M$ accepts it and, for any instance $x$ in $\overline{L}$ ($= \Sigma^*-L$), $M$ rejects it.
A function $f$ from $\Sigma^*$ to $\Gamma^*$ for two alphabets $\Sigma$ and $\Gamma$ is called \emph{logarithmic-space (or log-space) computable} if there exists a DTM equipped further with a write-once\footnote{A tape is \emph{write-once} if its tape head never moves to the left and, whenever the tape head writes a non-blank symbol, it must move to the next blank cell.} output tape that, on input $x\in\Sigma^*$, produces $f(x)$ on this output tape using $O(\log{|x|})$ work space. We write $\fl$ to denote the class of all polynomial-time log-space computable functions.

%%%

We briefly explain the notion of parameterized decision problems and  introduce the linear space hypothesis (LSH), discussed in \cite{Yam17a}.
A \emph{size parameter} $m$ over alphabet $\Sigma$ is a map from $\Sigma^*$ to $\nat$.
In particular, a \emph{log-space size parameter} refers to a size parameter that is computable in $n^{O(1)}$ time using $O(\log{n})$ space on any inputs of length $n$. A typical example of such a size parameter is $m_{\|}$ defined by $m_{\|}(x) = |x|$ for any string $x$.
A \emph{parameterized decision problem} has the form $(L,m)$ with a language $L$ over a certain alphabet $\Sigma$ and a size parameter $m$ over $\Sigma$.
Given such a parameterized decision problem $(L,m)$, we say that a DTM $M$   \emph{solves} $(L,m)$ \emph{in polynomial time using sublinear space} if $M$ solves $L$ and $M$ runs in time polynomial in $m(x)$ using space at most $m(x)^{\varepsilon}\ell(|x|)$ for a constant $\varepsilon\in[0,1)$ and polylogarithmic functions $\ell$ for all inputs $x$.
The complexity class $\psublin$, introduced in \cite{Yam17a}, is composed of all parameterized decision problems together with log-space size parameters solvable in polynomial time using sublinear space.

We often deal with Boolean formulas expressed in \emph{$k$-conjunctive normal form} ($k$CNF); that is, a conjunction of disjunctions of at most $k$ literals, where $k$ is a fixed positive number. The notation $\mathrm{2SAT}$ expresses the set of all satisfiable 2CNF Boolean formulas. For a fixed number $d\in\nat^{+}$,  we use the special notation of $\mathrm{2SAT}_d$ when 2CNF Boolean formulas are limited to the ones in which each variable appears at most $d$ times in the form of literals.
We write $m_{vbl}(\phi)$ and $m_{cls}(\phi)$ respectively for the total numbers of variables and of clauses used in a given Boolean formula $\phi$.
The \emph{linear space hypothesis} (LSH) is a statement that the parameterized decision problem $(\mathrm{2SAT}_3,m_{vbl})$ does not belong to $\psublin$.

In this work, however, we fix our size parameter $m$ used for $\psublin$ to the standard one $m_{\|}$  and we use the specific notation $\psublin_{\|}$ to denote the collection of all (standard) decision problems $L$ for which its parameterization $(L,m_{\|})$ belongs to $\psublin$.
The following proposition allows us to focus on the standard size parameter $m_{\|}$ when discussing LSH.

\begin{proposition}
LSH holds if and only if $\twosat_3\notin \psublin_{\|}$.
\end{proposition}

\begin{proof}
Given any 2CNF Boolean formula $\phi$, we write $V_{\phi}$ and $C_{\phi}$ for the set of all variables in $\phi$ and the set of all clauses in $\phi$, respectively. Notice that $m_{vbl}(\phi)=|V_{\phi}|$ and $m_{cls}(\phi)=|C_{\phi}|$. Moreover, we write $|\phi|$ for the length of binary representation of $\phi$, assuming the existence of a reasonable binary encoding of Boolean formulas.
Let us recall that LSH holds iff $(\twosat_3,m_{vbl})\notin \psublin$.
Assume first that LSH is false. We then take a DTM $M$ solving $(\twosat_3,m_{vbl})$ in time polynomial in $m_{vbl}(\phi)$ using space at most $m_{vbl}(\phi)^{\varepsilon}\ell(|\phi|)$ for a constant $\varepsilon\in[0,1)$ and a polylog function $\ell$.
Clear, $\max\{m_{vbl}(\phi),m_{cls}(\phi)\} \leq |\phi|$ follows. We then obtain  $m_{vbl}(\phi)\leq 2m_{cls}(\phi)\leq 2|\phi|$ and $m_{vbl}(\phi)^{\varepsilon}\ell(|\phi|)\leq (2|\phi|)^{\varepsilon}\ell(|\phi|)\leq |\phi|^{\varepsilon}\ell'(|\phi|)$, where $\ell'(n) = 2\ell(n)$. From this, we conclude that $M$ solves $\twosat_3$ in $|\phi|^{O(1)}$ time using at most $|\phi|^{\varepsilon}\ell'(|\phi|)$ space. This implies that $\twosat_3\in\psublin_{\|}$.

Conversely, assume that $\twosat_3\in \psublin_{\|}$.
It follows that $|\phi| \leq d |C_{\phi}|\log|V_{\phi}|$ for a certain constant $d>0$, which is independent of the choice of $\phi$. From this fact, we conclude that $|\phi|\leq m_{cls}(\phi)\log{m_{vbl}(\phi)} \leq m_{cls}(\phi) \log2m_{cls}(\phi)$. This implies that $M$ runs in time $m_{cls}(\phi)^{O(1)}$ using space at most $m_{cls}(\phi)^{\varepsilon}\ell'(m_{cls}(\phi))$ for an appropriate polylog function $\ell$. It then follows that $(\twosat_3,m_{cls})$ belongs to $\psublin$. It is known in \cite{Yam17a} that $(\twosat_3,m_{vbl})\in \psublin$ iff $(\twosat_3,m_{cls})\in\psublin$. Therefore, LSH does not hold.
\end{proof}

%%%
%%%

Given two decision problems $L_1$ and $L_2$, we say that $L_1$ is \emph{$\dl$-m-reducible to} $L_2$ (denoted $L_1\Lreduces L_2$) if there is a function $f$ in $\fl$ such that, for all $x$, $x\in L_1$ holds exactly when $f(x)\in L_2$. Moreover, $L_1$ is said to be \emph{$\dl$-m-equivalent to} $L_2$ (denoted $L_1\Lequiv L_2$) if $L_1\Lreduces L_2$ and $L_2\Lreduces L_1$ both hold.

%%%
\subsection{Syntactic NL (or SNL)}\label{sec:syntactic-NL}

Let us explain the fundamental terminology given in \cite{Yam17a}.
Although the original concepts were introduced in a discussion on the computational complexity of ``parameterized'' decision problems, in this work, we fix $m_{\|}(x)=|x|$ as our size parameter and we modify these concepts to fit in the setting of ``standard'' (i.e., non-parameterized) decision problems. Whenever we need to refer to the original ``parameterized''  $\snl$, we emphatically write $\para\snl$, as in \cite{Yam17a}, to avoid any confusion of the reader.

As an introduction of the syntax of our logical system, we start with explaining \emph{syntactic NL sentences} (or \emph{SNL sentences}, for short).

\begin{definition}
A \emph{vocabulary} (an \emph{input signature} or an \emph{input relation})  is a finite set composed of $(S_i,k_i)$, $c_j$, $0$, $n$, $suc$, $pred$ for all $i\in[d]$ and $j\in[d']$, where $S_i$ is a
\emph{predicate symbol} of arity $k_i\geq0$ (or a $k_i$-arity predicate symbol), $0$ and $n$ are \emph{constant symbols}, $c_j$ is another specific symbol expressing an ``input object'' (such as a number, a vertex or an edge of a graph, and a size of columns or rows of a matrix) of the target computational problem, and $pred$ and $suc$ are
two designated function symbols called respectively by the \emph{successor function} and the \emph{predecessor function}.
\end{definition}

The meanings of $suc(\cdot)$ and $pred(\cdot)$ are, as their names suggest,  $suc(i)=i+1$ and $pred(i)=\max\{0,i-1\}$ for any $i\in\nat$. We often abbreviate $suc(suc(i))$ as $suc^2(i)$ and $suc(suc^2(i))$ as $suc^3(i)$, etc. We further write $i+e$ for $suc^e(i)$ when $e$ is a constant in $\nat^{+}$.

To types of variables were used in \cite{Yam17a}. \emph{First-order variables}, denoted by $i,j,\ldots, u,v,\ldots$, range over all natural numbers and input objects  (such as vertices or edges of a graph and entries of a matrix) used to describe various parts of an instance of a target computational problem.
\emph{Second-order variables}\footnote{In \cite{Yam17a}, a second-order variable is limited to have only two argument places. To make it more general, we here allow the second-order variable to have more than two argument places.},
denoted by $P,Q,\ldots$, in this work range over a specific form of  \emph{relations} whose first argument takes a natural number and the other arguments take multiple input objects. This first-argument restriction of  second-order variables is necessary to ``express'' log-space computing. See \cite{Yam17a} for more information.

\emph{Terms} include first-order variables, constant symbols, and function symbols.
An \emph{atomic formula} has one of the following forms:  $S_j(u_1,\ldots,u_{k_i})$, $P(i,v_1,v_2,\ldots,v_k)$, $u=v$, and $i\leq j$, where $i,j,u,v,u_1,\ldots,u_{k_i},v_1,\ldots,v_k$ are terms, and $P$ is a second-order variable.
For clarity reason, we use $i,j$ for number-related terms and $u,v,u_1,\ldots,u_{k_i},v_1,\ldots,v_k$ for terms associated with other input objects.
\emph{Formulas} are built inductively from atomic formulas by connecting them with logical connectives ($\to$, $\neg$, $\vee$, $\wedge$) and first/second-order quantifiers ($\forall$, $\exists$). If a formula contains no free variables, then it is called a \emph{sentence}.
Notice that $\to$ and $pred$ are included here for our convenience although they are redundant because $\phi\to\psi$ is equivalent to $\neg\phi\vee \psi$, and $pred(i)=j$ is equivalent to $i=j=0\vee suc(j)=i$.

In this work, we concentrate on the specific case where second-order variables represent only ``functions''. It is therefore convenient to introduce a functional variant of the second-order quantifier.
For this purpose, we use the special notation $\exists^fP[\psi(P)]$ with a formula $\psi$ containing no second-order quantifiers as a shorthand for $\exists P[\psi(P) \wedge Func(P)]$, where
$Func(P)$ is a unique sentence over a second-order variable $P$ expressing that $P(\cdot,\cdot)$ works as a ``function''; namely, $Func(P) \equiv Func_1(P)\wedge Func_2(P)$, where $Func_1(P) \equiv (\forall i) (\exists w_1,\ldots,w_k)[P(i,w_1,\ldots,w_k)]$ and  $Func_2(P) \equiv (\forall i, u_1,\ldots,u_k,v_1,\ldots,v_k) [P(i,u_1,\ldots,u_k)\wedge P(i,v_1,\ldots,v_k)\rightarrow \bigwedge_{j=1}^{k}(u_i=v_i)]$.
Here, the symbol ``$\exists^f$'' is emphatically called the \emph{functional existential quantifier} and the variable $P$ (in the scope of $\exists^f$) is called a \emph{functional variable}.

\begin{definition}\label{SNL-sentence}
Let $\VV=\{(S_i,k_i),c_j, 0, n, suc, pred \mid i\in[d],j\in[d']\}$ denote a  vocabulary. A \emph{syntactic NL sentence} (or an \emph{SNL sentence}) over $\VV$ is a second-order sentence $\Phi$ of the form:
\begin{eqnarray*}
\lefteqn{\Phi \equiv \exists^f P_1\cdots \exists^f P_l\: \forall i_1\cdots \forall i_r \: \forall y_1\cdots \forall y_s} \hs{10}\\
&& [ \bigwedge_{j=1}^{t} \psi_j(P_1,\ldots,P_l,i_1, \ldots,i_r,y_1, \ldots,y_s,S_1,\ldots, S_d,c_1,\ldots,c_{d'})],
\end{eqnarray*}
where $l,r,s,t\in\nat$ and each $\psi_j$ ($j\in[t]$) is a quantifier-free second-order formula for which no two $\psi_j$'s share any common first-order variables, where all variables are listed on the above expression of $\psi_j$ only for simplicity.
Here, $P_1,\ldots,P_l$ are second-order functional variables  $i_1,\ldots,i_r$ are first-order variables representing natural numbers, and $y_1,\ldots,y_s$ are also first-order variables representing all other input objects. Each formula $\psi_j$ should satisfy the following two \emph{second-order variable requirements}.
\begin{enumerate}\vs{-1}
  \setlength{\topsep}{-2mm}%
  \setlength{\itemsep}{0mm}%
  \setlength{\parskip}{0cm}%

\item[(i)] Each $\psi_j$ contains only second-order variables of the form {$P_k(i,\boldvec{v}_1), P_k(suc(i),\boldvec{v}_2), P_k(suc^2(i),\boldvec{v}_3), \ldots$, $P_k(suc^a(i),\boldvec{v}_{a+1})$} for a fixed constant $a\in\nat^{+}$, where each of $\boldvec{v}_1,\ldots,\boldvec{v}_{a+1}$ is a $k'$ tuple of terms for a fixed constant $k'\in\nat^{+}$.

\item[(ii)] $\psi_j$ can be rewritten in the logically-equivalent form of finite ``disjunctions''  satisfying the following condition: among those disjuncts, there are only at most two disjuncts containing second-order variables and each of them must have the form $(\bigwedge_{k,i,\boldvec{v}} P_k(i,\boldvec{v}))\wedge (\bigwedge_{k',i',\boldvec{v'}} \neg P_{k'}(i',\boldvec{v'})) \wedge R$, where $R$ is an appropriate subformula including no second-order variable.
\end{enumerate}
\end{definition}

The requirement (ii) was originally introduced as a natural analogue of the formulation of 2SAT. It is significantly important to guarantee that SNL (defined in Definition \ref{def-of-SNL}) is contained in NL.

Next, we explain the semantics of $\snl$ sentences.

\begin{definition}\label{SNL-model}
Let $\VV=\{(S_i,k_i),c_j, 0, n, suc, pred \mid i\in[d],j\in[d']\}$ denote any vocabulary.

(1) A \emph{relational structure} $\SSS$ over $\VV$ is a set of tuples $(U_i,D_i,k_i)$ and $(\underline{c}_j,V_j)$ with finite universes $U_i$ and $V_j$ of ``input objects'' (including natural numbers) and domains $D_i$ associated with predicate symbols $S_i$ in $\VV$ satisfying $D_i\subseteq U_i^{k_i}$, and constants $\underline{c}_j$ in $V_j$.
The constant symbols $c_j$ are interpreted as $\bar{c}_j$ and the predicate symbols $S_i$ are interpreted as $D_i$ so that, if input objects $\underline{s}_1,\underline{s}_2,\ldots,\underline{s}_{k_i}$ in $U_i$ are assigned respectively to variables $x_1,x_2,\ldots,x_{k_i}$ used for $S_i$, the formula $S_i(x_1,x_2,\ldots,x_{k_i})$ is true exactly when $(\underline{s}_1,\underline{s}_2,\ldots,\underline{s}_{k_i})\in D_i$.

(2) Let $\Phi$ denote any SNL sentence of the form of Definition \ref{SNL-sentence} with variables $P_1,\ldots,P_l,i_1,\ldots,i_r,y_1,\ldots,y_s$ over $\VV$.
A \emph{domain structure} $\DD$ for $\Phi$ is the union of three sets $\{(P_j,[0,e_j]_{\integer}\times U'_{j_1}\times\cdots \times U'_{j_{k'}},k'+1)\}_{j\in[l]}$, $\{(i_j,[0,e'_j]_{\integer})\}_{j\in[r]}$, and $\{(y_j,U''_j)\}_{j\in[s]}$, which provide the scopes of variables of $\Phi$ in the following manner for fixed constants $e_j,e'_j\in\nat^{+}$. Each second-order variable $P_j$ ($j\in[l]$) ranges over all elements in $[0,e_j]_{\integer} \times U'_{j_1}\times\cdots\times U'_{j_{k'}}$, each first-order variable $i_j$ ($j\in[r]$) ranges over all numbers in $[0,e'_j]_{\integer}$, and each variable $y_j$ ($j\in[s]$) ranges over all elements in $U''_{j}$.
\end{definition}

Concrete examples of relational and domain structures will be given in Examples \ref{2COLOR} and \ref{example-UK}.

%%%%

\begin{definition}
A relational structure $\SSS$ over vocabulary $\VV$ is said to \emph{describe} (or \emph{represent}) an instance $x$ of the target computational problem if every input object appearing in $x$ has either its corresponding predicate symbol in $\VV$ with its universe and domain in $\SSS$ or its corresponding  constant symbol in $\VV$ with its universe in $\SSS$.
\end{definition}

It is important to remark that, when $\SSS$ describes $x$, since the universes $U_i$ and $V_i$ must be completely specified inside $x$, their sizes $|U_i|$ and $|V_i|$ should be upper-bounded by $O(|x|)$.

When a relational structure $\SSS$ and a domain structure $\DD$ are given for an SNL sentence $\Phi$, it is possible to determine the \emph{validity} of $\Phi$ by interpreting all predicate symbols $S_i$ and all constant symbols $c_j$ appearing in $\Phi$ as domains $D_i$ and constants $\underline{c}_j$ in $\SSS$ and by assigning input objects in $\SSS$  and $\DD$ to variables appropriately.
This interpretation makes $\Phi$ either ``true'' or ``false''.
Notationally, we write $(\SSS,\DD)\models \Phi$ if $\Phi$ is true on $\SSS$ and $\DD$.
When $(\SSS,\DD)$ are clear from the context, nevertheless, we further omit $(\SSS,\DD)$ and simply write $\models \Phi$.

\begin{definition}\label{def-of-SNL}
Given a decision problem $A$ and an $\snl$ sentence $\Phi$ over vocabulary $\VV$, we say that $\Phi$ \emph{syntactically expresses} $A$ if, for any instance $x$ to $A$, there are a relational structure $\SSS_x$ over $\VV$  describing $x$ and a domain structure $\DD_x$ for $\Phi$ satisfying the following condition: $x\in A$ iff $\Phi$ is true on $\SSS_x$ and $\DD_x$.
\end{definition}

It is possible to view the syntactical expressibility by an $\snl$ sentence as a (non-probabilistic) Merlin-Arthur interactive proof system, in which Merlin provides a polynomial-size ``proof'' to Arthur who check the validity of the proof by a logarithmically space-bounded algorithm.

\begin{definition}
We denote by $\mathrm{SNL}$ the collection of all decision problems $A$ such that there exist a vocabulary $\VV$ and an SNL sentence $\Phi$ over $\VV$ for which $\Phi$ syntactically expresses $A$.
\end{definition}

In \cite{Yam17a}, the parameterized decision problem $(\dstcon,m_{ver})$ was shown to be in $\para\snl$, where the size parameter $m_{ver}$ indicates the total number of vertices of a given graph.
By reviewing the corresponding proof of this fact, we can conclude that $\dstcon$ belongs to (the ``non-parameterized'' class) $\snl$.
As another quick example, we see how to construct an $\snl$ sentence to express the decision problem $\mathrm{2COLOR}$, in which one asks whether a given undirected graph is colorable using only two colors. This is the same as checking that a given graph is bipartite.
It is known that $\mathrm{2COLOR}$ falls in $\dl$ (see \cite{AG00} with the fact that $\mathrm{SL}=\dl$).

\begin{yexample}\label{2COLOR}
{\rm
We wish to show that $\mathrm{2COLOR}$ belongs to $\snl$ by constructing an appropriate $\snl$ sentence for $\mathrm{2COLOR}$.
Given an instance $x$ of an undirected graph $G=(V_G,E_G)$, we assume that $V_G =\{v_1,v_2,\ldots,v_n\}$. We identify each vertex $v_i$ with the integer $i$.
Hence, $V_G$ is viewed as $[n]$ and $E_G$ is viewed as a subset of $[n]\times[n]$.
Let $E$ denote a predicate symbol associated with $E_G$.
We define $\VV=\{(E,2),0,1\}$  and  $\SSS_x=\{(U_x,D_x,2), (0,V_{1}), (1,V_{2})\}$ with $U_x=[n]$, $D_x=\{(i,j)\mid (x_i,x_j)\in \hat{E}\}$, $V_1=\{0\}$, and $V_2=\{1\}$. Clearly, $\SSS_x$ describes $x$.
Next, we define a sentence $\Phi$ to be $(\exists^{f}C) (\forall i, d, i',  j', d', e')  [\Phi_1(C,i,d) \wedge  \Phi_2(C,E,i',j',d',e') ]$, where $i,j,i',d,d',e'$ are symbols expressing the first-order variables ranging over $[n]$,   $\Phi_1(C,i,d) \equiv C(i,d)\to 0\leq d\leq 1$, and $\Phi_2(C,E,i',j',d',e')  \equiv  E(i',j')\wedge C(i',d')\wedge C(j',e')\to d'\neq e'$. The sentence $\Phi$ informally asserts that, for an appropriate coloring of vertices, (i) we use only two colors and (ii) two endpoints of each edge are colored by distinct colors.
By rewriting $\Phi_1$ and $\Phi_2$ in the disjunction form, we can show that $\Phi_1$ and $\Phi_2$ both satisfy the second-order variable requirements.
To see this fact, we note that $\Phi_1$ is logically equivalent to $\neg C(i,d)\vee 0\leq d\leq 1$ and that $\Phi_2$ is to $\neg E(i',j')\vee C(i',d')\vee C(j',e') \vee d'\neq e'$. The two last formulas show that the required conditions on the second-order variables are clearly satisfied.
We further define $\DD_x=\{(C,[n]\times \{0,1\},2)\}\cup
\{ (s,[n])\mid s\in\{i,i',j'\}\}  \cup \{(s',[2]) \mid s'\in\{d,d',e'\}\}$.
It then follows that $\Phi$ is true on $\SSS_x$ and $\DD_x$ iff $x\in \mathrm{2COLOR}$.
}
\end{yexample}

Another example is the \emph{unary 0-1 knapsack problem} (UK), which was discussed by Cook \cite{Coo85}. An instance of UK is a series $(1^b,1^{a_1},1^{a_2},\ldots,1^{a_n})$ of unary strings with $b,a_1,a_2,\ldots,a_n\in\nat^{+}$ and one asks to determine the existence of a subset $S$ of $[n]$ satisfying $\sum_{i\in S}a_i=b$.
It was shown in \cite{Yam23} that $\mathrm{UK}$ belongs to a subclass of $\nl$, called $\mathrm{1t1NCA}$ (see \cite{Yam23} for details).

\begin{yexample}\label{example-UK}
{\rm
We claim that the decision problem $\mathrm{UK}$ is also in $\mathrm{SNL}$. To see this, let $x=(1^b,1^{a_1},1^{a_2},\ldots,1^{a_n})$ denote any instance given to $ \mathrm{UK}$. For simplicity, we assume that $a_i\leq b$ for all indices $i\in[n]$.
We then prepare two predicate symbols $I$ and $ADD$ for which $I(i,a)$ means that $a$ is the $i$th input value $a_i$ of $x$ and $ADD(c,a,b)$ means $c=a+b$. We set $\Phi\equiv (\exists^fP)(\forall i,s,t) [P(0,0)\wedge P(n,b)\wedge (\psi_1(P,i,s,t)\to \psi_2(P,I,ADD,s,t))]$, where $\psi_1\equiv i<n \wedge P(i,s)\wedge P(i+1,t)$ and $\psi_2\equiv  s=t\leq b\vee (s<t\leq b\wedge (\forall z)[I(i+1,z)\wedge z>0\to ADD(t,s,z)])$.
Notice that the formula $\psi_1\to\psi_2$ satisfies the second-order variable requirements since $\psi_1\to\psi_2$ can be rewritten as $\neg P(i,s)\vee \neg P(i+1,t)\vee R$ for an appropriate formula $R$ containing no second-order variables. We then conclude that $\Phi$ is an $\snl$ sentence.

We set $\VV=\{(I,2),(ADD,3),0,n,b\}$, where $b$ is treated as a constant. We define $\SSS_x=\{(U_x,D_I,2), (U_x,D_{ADD},3), (\bar{b},U_x)\}$, where $U_x=[0,b]_{\integer}$, $D_I = \{(\underline{i},\underline{a})\mid \underline{i}\in[0,n]_{\integer}, \underline{a}\in U_x\}$, and $D_{ADD}=\{(\underline{t},\underline{s},\underline{z})\mid \underline{t},\underline{s},\underline{z}\in U_x, \underline{t}=\underline{s}+\underline{z} \}$. Moreover, we set $\DD_x=\{(P,U_n\times U_x,2)\} \cup \{(i,U_n)\} \cup \{(u,U_x)\mid u\in\{s,t,z\}\}$ with $U_n=[0,n]_{\integer}$.
It then follows that $x=(1^b,1^{a_1},1^{a_2},\ldots,1^{a_n})\in \mathrm{UK}$ iff $\Phi$ is true on $\SSS_x$ and $\DD_x$.
}
\end{yexample}

%%%%

A subclass of $\para\snl$, which is called $\para\snl_{\omega}$, was also  introduced in \cite{Yam17a}. This subclass has a direct connection to the linear space hypothesis (LSH). Here, we introduce its ``non-parameterized'' version as follows.

\begin{definition}\label{def:SNL-omega}
The complexity class $\mathrm{SNL}_{\omega}$ is composed of all decision problems $A$ in $\mathrm{SNL}$ that enjoys the following extra requirements. Let $\Phi$ denote any SNL-sentence of the form given in Definition \ref{SNL-sentence} with $t$ quantifier-free subformulas $\psi_j(P_1,\ldots,P_l,\boldvec{i}, \boldvec{y}, S_1,\ldots,S_d,c_1, \ldots,c_{d'})$ together with (hidden) sentence $Func(P_i)$ for all $i\in[l]$, where $\boldvec{i}= (i_1, \ldots,i_r)$ and $\boldvec{y}=(y_1, \ldots,y_s)$.
Assume that $\Phi$ syntactically expresses $A$ by a certain relational structure $\SSS_x$ and a certain domain structure $\DD_x$ associated with each instance $x$ given to $A$.
Here, each $\psi_j$ must satisfy the second-order variable requirements.
We further demand that the sentence $(\bigwedge_{h=1}^{l} Func(P_h))$ must be  ``expressed'' inside $\Phi$ with no use of existential quantifiers ``$\exists$'' in the following sense: $\exists P_1\cdots \exists P_l \forall \boldvec{i}\forall \boldvec{y} [(\bigwedge_{j=1}^{t}\psi_j) \wedge (\bigwedge_{h=1}^{l} Func(P_h))]$ is true iff $\exists P_1\cdots \exists P_l \forall \boldvec{i}
\forall \boldvec{y} [\bigwedge_{j=1}^{l}\psi_j]$ is true, where ``$\forall \boldvec{i}$'' and ``$\forall\boldvec{y}$''  are respectively shorthands for $\forall i_1\forall i_2\cdots \forall i_r$ and $\forall y_1\forall y_2\cdots \forall y_s$.
\end{definition}

As in \cite{Yam17a}, let us consider $\mathrm{exact}3\dstcon$, which is the directed $s$-$t$ connectivity problem restricted to directed graphs of degree exactly $3$. In the parameterized setting, it was shown in \cite{Yam17a}, the parameterized decision problem $(\mathrm{exact}3\dstcon,m_{ver})$ belongs to $\para\snl_{\omega}$. In essence, a similar argument leads to the claim that $\mathrm{exact}2\dstcon$ is in $\snl_{\omega}$.

%%%%%%%%
%%%%%%%%
\section{Structural Properties of SNL}\label{sec:structure}

Through Section \ref{sec:notions}, we have reviewed from \cite{Yam17a} the logical notion of \emph{$\snl$ sentences} and the associated complexity class \emph{SNL}. In succession to the previous section, we intend to study the structural properties of $\snl$ in depth. In particular, we are focused on the closure properties of SNL under Boolean operations.

%%%
\subsection{Basic Closure Properties and L-m-Reductions}\label{sec:basic-closure}

It is known that $\nl$ is closed under union, intersection, and complementation. Similarly, $\snl$ enjoys the closure properties under union and intersection.

\begin{proposition}\label{intersection-union}
$\snl$ is closed under union and intersection.
\end{proposition}

\begin{yproof}
Let $A$ and $B$ denote two arbitrary decision problems in $\snl$. Take $\snl$-sentences $\Phi_A$ and $\Phi_B$ that syntactically express $A$ and $B$, respectively. Since $\Phi_A$ and $\Phi_B$ are $\snl$ sentences, we assume that $\Phi_A\equiv \exists^f\boldvec{P}  \forall \boldvec{i}\forall \boldvec{y} [\bigwedge_{j=1}^{t} \psi_j]$ and $\Phi_B\equiv \exists^f\boldvec{R} \forall \boldvec{i'}\forall \boldvec{y'} [\bigwedge_{j'=1}^{s} \xi_{j'}]$, where $\boldvec{P}=(P_1,P_2,\ldots, P_l)$, $\boldvec{R}=(R_1,R_2,\ldots,R_{l'})$, and $\psi_j$'s and $\xi_{j'}$'s are all quantifier-free formulas and that $\Phi_{A}$ and $\Phi_{B}$ satisfy the second-order variable requirements.
Assume further that each $\psi_k$ has the form $\bigvee_{k_j}\hat{\psi}_{k_j}$ and each $\xi_{j'}$ has the form $\bigvee_{l_{j'}}\hat{\xi}_{l_{j'}}$.
For simplicity, all elements in $(\boldvec{P}, \boldvec{i},\boldvec{y})$ and $(\boldvec{R}, \boldvec{i'},\boldvec{y'})$ do not share any common variables.

For the target intersection $C=A\cap B$, we define $\Phi \equiv \Phi_A\wedge \Phi_B$, which is logically equivalent to $\exists^f\boldvec{P} \exists^f \boldvec{R} \forall \boldvec{i} \forall  \boldvec{i'} \forall \boldvec{y} \forall  \boldvec{y'} [(\bigwedge_{j} \psi_j)\wedge (\bigwedge_{j'} \xi_{j'})]$.
Since all $\psi_j$'s and $\xi_{j'}$'s satisfy the second-order variable requirements, so does the formula $\phi\equiv (\bigwedge_{j} \psi_j)\wedge (\bigwedge_{j'} \xi_{j'})$. Hence, $\Phi$ is also an $\snl$ sentence. By definition, $\Phi$ syntactically expresses $C$.

For the case of union, $C'=A\cup B$, we cannot simply define a sentence $\Phi'$ as $\Phi'\equiv \Phi_A\vee \Phi_B$ using $\Phi_A$ and $\Phi_B$. Instead, we need to define $\Phi'$ as follows. Let us introduce a new variable $k$, which is assumed to take either $1$ or $2$. This $k$ is intended to indicate which of $\Phi_A$ and $\Phi_B$ is true.
Let us first define  $\Xi_1\equiv k=1\to \bigwedge_{j}\psi_j$ and $\Xi_2\equiv k=2\to \bigwedge_{j'}\xi_{j'}$. We then define $\Phi' \equiv  \exists^f\boldvec{P} \exists^f\boldvec{R} \forall\boldvec{i} \forall\boldvec{y}\forall\boldvec{i'}\forall\boldvec{y'}\forall k[1\leq k\leq 2 \to \Xi_1\wedge \Xi_2]$. Notice that $\Xi_1$ and $\Xi_2$ are rephrased as $\Xi_1\equiv \bigwedge_{j}(\neg(1\leq k\leq 2)\vee k\neq1\vee \psi_j)$ and $\Xi_2\equiv \bigwedge_{j'}(\neg(1\leq k\leq 2)\vee k\neq2\vee \xi_{j'})$.
Clearly, $\Xi$ is logically equivalent to the conjunction of $\Xi_1$ and $\Xi_2$.
It is not difficult to check that the formula $1\leq k\leq 2 \to \Xi_1\wedge \Xi_2$ satisfies the second-order variable requirements. It thus follows that $\Phi'$ syntactically expresses $C'$.
\end{yproof}

%%%

Given a decision problem $A$, the notation $\Lreduces\!\!(A)$ expresses the collection of all decision problems that are $\dl$-m-reducible to $A$. Furthermore, for a given complexity class $\CC$, $\Lreduces\!\!(\CC)$ denotes the union $\bigcup_{A\in \CC} \Lreduces\!\!(A)$.
Since $\nl$ is closed under $\dl$-m-reductions, $\Lreduces\!\!(\nl)=\nl$ follows.
Concerning SNL, we obtain the following.

\begin{proposition}\label{SNL-closure-NL}
$\nl = \;
\Lreduces\!(\snl_{\omega}) =\;
\Lreduces\!\!(\snl)$.
\end{proposition}

\begin{yproof}
By definition, $\snl_{\omega}\subseteq \snl$ follows. We thus obtain $\Lreduces\!\!(\snl_{\omega}) \subseteq\; \Lreduces\!\!(\snl)$.

In the ``non-parameterized'' setting of this work, it is possible to rephrase this inclusion
as $\snl\subseteq \nl$ by restricting corresponding size parameters to $m_{\|}$.
Therefore, we obtain 
$\Lreduces\!\!(\snl)\subseteq \; \Lreduces\!\!(\nl) =\nl$.

In the parameterized setting, it is shown in \cite{Yam17a} that $(\mathrm{exact3}\dstcon,m_{ver})$  is complete for $\para\snl_{\omega}$ under so-called \emph{short SLRF-T-reductions} and it is also \emph{short L-m-reducible} to $(3\dstcon,m_{ver})$.
It is also shown in \cite{Yam17a} that the decision problem $\mathrm{exact3}\dstcon$ is complete for $\nl$ under $\dl$-m-reductions. These results together imply that $\nl \subseteq \; \Lreduces\!\!(\snl_{\omega})$.
\end{yproof}

%%%%%

Actually, we can assert a stronger statement than Proposition \ref{SNL-closure-NL}. Here, we intend to claim that $\snl$ occupies a ``structurally'' important portion  of $\nl$ in the sense described in the following theorem.

\begin{theorem}\label{equivalent-MonoSNL}
For any decision problem in $\nl$, there always exists its $\dl$-m-equivalent problem in $\snl$.
\end{theorem}

\begin{yproof}
It is known that all decision problems in $\nl$ are solvable  by appropriate \emph{4-counter two-way nondeterministic  counter automata}\footnote{This computation model is also known as \emph{counter machines}, where a \emph{counter} refers to a stack with a single stack symbol except for the bottom marker $\bot$.} (2ncta's) in polynomial time. See, e.g., \cite[Proposition 2.3]{Yam23} for the proof of this fact.

Let $L$ denote an arbitrary decision problem in $\nl$ and take a 4-counter 2ncta $M$ of the form $(Q,\Sigma,\{1\},\{\triangleright,\triangleleft\}, \delta,q_0, \bot, Q_{acc},Q_{rej})$ that solves $L$ in polynomial time. Note that $\delta$ maps $(Q-Q_{halt})\times \check{\Sigma}_{\lambda} \times \{1,\bot\}^4$ to $Q\times D \times (\{1\}^*\cup\{\varepsilon\})^4$, where $\check{\Sigma}_{\lambda} = \Sigma\cup\{\lambda,\triangleright,\triangleleft\}$, $D=\{-1,+1\}$ (tape head directions) and $Q_{ha,t} = Q_{acc}\cup Q_{rej}$.
To ease the description of the following construction, $M$ is assumed to halt exactly in $n^k$ steps (for an appropriate constant $k\in\nat^{+}$) with the empty counters (except for $\bot$). Moreover, we assume that $Q_{acc}=\{q_{acc}\}$ and that $M$ takes exactly two nondeterministic choices at any step (i.e., $|\delta(q,l,\boldvec{a})|=2$ for any $(q,l,\boldvec{a})$).

Let us consider the decision problem $H\!ALT_{M}$, in which, for any given instance $x$, we must determine whether there exists an accepting computation path of $M$ on $x$. In what follows, we wish to show that $H\!ALT_{M}$ belongs to $\snl$.

Hereafter, we fix an instance $x$ arbitrarily and intend to express an accepting computation path of $M$ on $x$.
A \emph{configuration} of $M$ on $x$ is of the form $(q,l,\boldvec{w})$ with $q\in Q$, $l\in[0,|x|+1]_{\integer}$, and $\boldvec{w}=(w_1,w_2,w_3,w_4)\in(\{1\}^*\bot)^4$.
This means that $M$ is in inner state $q$, scanning the $l$th tape cell with the $i$th counter holding $w_i$ for any $i\in[4]$.
For two configurations $(q,l,\boldvec{w})$ and $(p,m,\boldvec{v})$,
we write $(q,l,\boldvec{w})\vdash (p,m,\boldvec{v})$ if $M$ changes  $(q,l,\boldvec{w})$ to $(p,m,\boldvec{v})$ in a single step.
To describe a transition, we prepare three predicate symbols, $Top$, $Chan$, and $Delt$, whose intended meanings are given as follows.
(i) $Top(l,\boldvec{w},c,\boldvec{a})$ is true iff $c = x_{(l)}$ and $\boldvec{a}$ is top symbols of the counters,
(ii) $Chan(\boldvec{w},\boldvec{b},\boldvec{v})$ is true iff $\boldvec{w}$ is changed to $\boldvec{v}$ by modifying the top symbols of $\boldvec{w}$ to $\boldvec{b}$ by applying $\delta$, and
(iii) $Delt(q,c,\boldvec{a},p,d,\boldvec{b})$ is true iff $(p,d,\boldvec{b})\in \delta(q,c,\boldvec{a})$.

We also prepare a second-order variable $P$ so that $P(i,q,l,\boldvec{w})$ is true iff  $(q,l,\boldvec{w})$ is a configuration at time $i$.
We then define $\Phi$ to be $(\exists^fP) (\forall u,u',z,p,q,l,c,d,\boldvec{w},\boldvec{a},\boldvec{b})[\Phi_1\wedge \Phi_2]$, where
$\Phi_1 \equiv 0\leq i<n^k \wedge P(i,q,l,\boldvec{w})\wedge P(i+1,p,l+d,\boldvec{v}) \wedge Top(l,\boldvec{w},c,\boldvec{a})
\to \bigvee_{(p,d,\boldvec{b})\in \delta(q,l,\boldvec{a})} (Delt(q,c,\boldvec{a},p,d,\boldvec{b})  \wedge Chan(\boldvec{w},\boldvec{b},\boldvec{v}))$  and
$\Phi_2\equiv P(n^k,q,l,\boldvec{w})\to (q,l,\boldvec{w}) = (q_{acc},n,\bot,\ldots,\bot)$.
This formula $\Phi$ is clearly  an $\snl$ sentence.
It then follows by definition that $\Phi$ is true iff $M$ has an accepting computation path on $x$ iff $x\in HALT_A$. Therefore, $\Phi$ syntactically expresses $HALT_A$.
\end{yproof}

Another consequence of Theorem \ref{equivalent-MonoSNL} is the following statement concerning the so-called \emph{dichotomy theorem}, which asserts that every problem in $\snl$ is either in $\dl$ or $\nl$-complete.

\begin{corollary}\label{dichotomy-L-NL}
If $\dl\neq\nl$, then the dichotomy theorem does not hold for $\mathrm{SNL}$.
\end{corollary}

\begin{yproof}
It was shown in \cite{Yam17a} that, under the assumption of $\dl\neq\nl$, there are an infinite number of $\Lequiv$-equivalent classes within $\nl$ in the setting of parameterized problems. This fact can be easily translated into the non-parameterized version.
By Theorem \ref{equivalent-MonoSNL}, we can conclude that $\mathrm{SNL}$ also contains an infinite number of $\Lequiv$-equivalent classes. Thus, the dichotomy theorem does not hold for $\mathrm{SNL}$.
\end{yproof}

%%%
\subsection{Complementary Problems and $\mu$SNL}\label{sec:mu-operator}

We have demonstrated in Proposition \ref{intersection-union} that $\mathrm{SNL}$ is closed under union and intersection.
Now, we wonder if, for any decision problem in $\snl$, its complementary problem also falls in $\snl$. As a simple example, we consider  $\mathrm{2COLOR}$, which is logically equivalent to checking whether a given undirected graph is bipartite.
Let us consider its ``complementary'' decision problem, known as the \emph{non-bipartite graph problem} (NBG), in which one asks to determine whether a given undirected graph is not bipartite.
We show that not only $\mathrm{2COLOR}$ but also $\mathrm{NBG}$ are expressible by appropriate $\snl$ sentences.

\begin{proposition}\label{NBG-mu}
$\mathrm{2COLOR}$ and its complementary problem $\mathrm{NBG}$ are both in $\snl$.
\end{proposition}

\begin{yproof}
Given an undirected graph $G=(V_G,E_G)$, we prepare two predicate symbols, $E$ and $ODD$, where $E$ corresponds to $E_G$ and $ODD(k)$ indicates that $k$ is an odd number. We then define $\Phi\equiv (\exists^fP)(\forall i,k,l)(\forall u,v)[1\leq i<n\wedge (\bigwedge_{m=1}^{5}\Phi_m)]$, where $\Phi_1\equiv P(1,u,k)\to k=1$, $\Phi_2\equiv P(i,u,k)\wedge P(i+1,v,l)\to ((l=k\wedge u=v)\vee (l=k+1\wedge u\neq v))$, $\Phi_3\equiv P(i,u,k)\wedge P(i+1,v,k+1)\to E(u,v)$, $\Phi_4\equiv P(1,u,k)\wedge P(n,v,l) \to u=v$, and $\Phi_5\equiv P(n,u,k)\to ODD(k)$. By definition, $\Phi$ syntactically expresses NBG.
\end{yproof}

Next, let us consider the decision problem $\dstcon$, which is known to be in $\snl$ \cite{Yam17a}, and its complementary decision problem, called $\mathrm{DSTNCON}$, in which one asks to determine whether, given a directed graph $G$ and two vertices $s$ and $t$, no path exists in $G$ from $s$ to $t$.
Since $\nl$ is closed under complementation \cite{Imm88,Sze88},  $\mathrm{DSTNCON}$ belongs to $\nl$.
Is it true that $\mathrm{DSTNCON}$ belongs to $\snl$ as well?
Although we know that there exists its $\dl$-m-equivalent problem in $\snl$ by Theorem \ref{equivalent-MonoSNL}, it is not clear that $\mathrm{DSTNCON}$ itself falls in $\snl$.

To tackle this question, we intend to expand the complexity class $\mathrm{SNL}$ by introducing the additional \emph{$\mu$-operator} applied to second-order variables. In the definition of SNL, second-order variables are treated as functional variables, indicating ``functions'' from natural numbers to tuples of various objects.
Given a second-order variable $P$, a \emph{$\mu$-term} is of the form $\mu z.P(i,z)$  indicating a ``unique''  object $z$ satisfying $P(i,z)$ for a given number $i$.
However, we do \emph{not} allow any nested application of the $\mu$-operator, such as $\mu y.Q(i,y, \mu z.P(i,z))$ for two second-order variables $P$ and $Q$.
This new term allows us to write, e.g., $P(i+1,\mu z.P(i,z)+2)$ in order to mean that $(\forall z) [P(i,z)\to P(i+1,z+2)]$ as well as $(\exists z) [P(i,z)\wedge P(i+1,z+2)]$ (because $P$ indicates a function) by eliminating any use of quantifiers associated with $z$.
Notice that $P(i,z)$ and $P(i+1,z+2)$ satisfy the second-order variable requirement (i). For this perspective, we do not allow, e.g., $P(\mu z.P(i,z),y)$ because $P(z,y)$ and $P(i,z)$ may not in general satisfy the requirement (i).

\begin{definition}\label{def-mu-term}
A \emph{$\mu$-term} has the form $\mu z.P(i,z)$ for a second-order functional variable $P$ with the following requirement: there is no nested application of the $\mu$-operator.
\end{definition}

\begin{definition}
We naturally expand $\snl$ sentences by including $\mu$-terms obtained with no use of the nested $\mu$-operator and by demanding that each formula $\psi_j$ in Definition  \ref{SNL-sentence} must contain at most one $\mu$-term.
It is important to remark that each sentence with $\mu$-terms must satisfy the second-order requirements (i)--(ii). We further demand that
(iii) $\mu$-terms are not exempt from the requirement (i). This means that, for example, in a subformula $P_k(i,\mu z.P_{k'}(i',z))$, the inequality $|i-i'|\leq a$ holds for a fixed constant $a$.
These three requirements (i)--(iii) are briefly referred to as the \emph{$\mu$-term requirements}.
The sentences in this expanded logical system of $\snl$ with $\mu$-terms are succinctly called \emph{$\mu$SNL sentences}. All decision problems syntactically expressed by those $\mu\snl$ sentences form the complexity class $\mu\snl$.
\end{definition}

\begin{lemma}
$\snl \subseteq \mu\snl \subseteq \nl$.
\end{lemma}

\begin{yproof}
Since $\snl\subseteq \mu\snl$ is obvious, we only need to show that $\mu\snl\subseteq \nl$.
The following proof is a loose extension of the one for \cite[Proposition 4.11]{Yam17a}, in which every ``parameterized'' decision problem in $\para\snl_{\omega}$ is short SLRF-T-reducible to $(\mathrm{2SAT},m_{vbl})$.
In a similar strategy, we wish to prove this statement.
Let us take an arbitrary language $L$ in $\mu\snl$ and consider a $\mu\snl$ sentence $\Phi$ of the form $\exists^f P_1\cdots \exists^f P_l \forall\boldvec{i}\forall \boldvec{y} [ \bigwedge_{j=1}^{t} \psi_j(P_1,\ldots,P_k,\boldvec{i},\boldvec{y})]$ (similarly to Definition \ref{SNL-sentence}) that syntactically expresses $L$.
Let us construct a nondeterministic Turing machine (or an NTM) $M$ for solving $L$ in polynomial time using only log space.
The intended machine starts with an input $x$.

By the second-order variable requirements (i)--(ii), each $\psi_j$ is expressed as finite disjunctions such that at most two disjuncts have the second-order variables and have the form $((\bigwedge_{k_1,i_1,\mathbf{\it v}}  P_{k_1}(i_1,\boldvec{v}))\wedge (\bigwedge_{k_2,i_2,\mathbf{\it v}'} \neg P_{k_2}(i_2,\boldvec{v}'))\wedge
(\bigwedge_{k_3,i_3} P_{k_3}(i_3,\mu z.P_{k_4}(i_4,z))) \wedge (\bigwedge_{k_5,i_5,\mathbf{\it x}} Q(\boldvec{x},\mu z.P_{k_5}(i_5,z))) \wedge R$, where $R$ contains no second-order variables or $\mu$-terms and both $i_3$ and $i'_3$ are within a distance of a fixed constant $a$ (as in Definition \ref{SNL-sentence}).
We rephrase $\psi_j$ properly as finite conjunctions such that, if any conjunct of them contains a second-order variable, it must contain one of the following formulas: $P_{k_1}(i_1,\boldvec{v}_1)$, $\neg P_{i_1}(i_1,\boldvec{v}_1)$, $P_{k_3}(i_3,\mu z.P_{k_4}(i_4,z))$, $\neg P_{k_3}(i_3,\mu z.P_{k_4}(i_4,z))$,
$Q(\boldvec{x},\mu z.P_{k_5}(i_5,z))$,
$P_{k_1}(i_1,\boldvec{v}_1)\vee P_{k_2}(i_2,\boldvec{v}_2)$,
$P_{k_1}(i_1,\boldvec{v}_1)\vee \neg P_{k_2}(i_2,\boldvec{v}_2)$,
$\neg P_{k_1}(i_1,\boldvec{v}_1)\vee \neg P_{k_2}(i_2,\boldvec{v}_2)$, $P_{k_1}(i_1,\boldvec{v}_1)\vee P_{k_3}(i_3,\mu z.P_{k_4}(i_4,z))$, and
$\neg P_{k_1}(i_1,\boldvec{v}_1)\vee P_{k_3}(i_3,\mu z.P_{k_4}(i_4,z))$.
Here, let us focus on $P_{k_1}(i_1,\boldvec{v}_1)\vee P_{k_3}(i_3,\mu z.P_{k_4}(i_4,z))$. Note that, by the $\mu$-term requirement (iii), for instance,  $|i_3-i_4|\leq a$ follows from $P_{k_3}(i_3,\mu z.P_{k_4}(i_4,z))$, where $a$ is a fixed constant.

The machine $M$ nondeterministically guesses $(i,\boldvec{v})$, which corresponds to $P_k(i,\boldvec{v})$. To evaluate the truth value of $P_{k_3}(i_3,\mu z.P_{i_4}(i_4,z))$, we need to obtain $z$ for which $P_{k_4}(i_4,z)$ with index $i_4$ and move to index $i_3$ (within a distance of $a$) to check if $P_{k_3}(i_3,z)$ is true.
Obviously, we need only $O(\log{n})$ space to remember each value $z$ used in $\mu z.P_{k_4}(i_4,z)$. Therefore, this entire process can be done using $O(\log{n})$ space.
\end{yproof}

Immerman \cite{Imm88} and Szelcepsc\'{e}nyi \cite{Sze88} proved that $\nl$ is closed under complementation.
Their proofs utilize an algorithmic technique known as \emph{inductive counting}, which employs the following abstract argument to prove that $\mathrm{DSTNCON}$ is in $\nl$.
Given a graph $G=(V,E)$, we inductively determine the number $N_i$ of vertices that are reachable from a given vertex $s$ within $i$ steps for any number $i\in[0,|V|]_{\integer}$. It is possible to calculate $N_{i+1}$ from $N_i$ nondeterministically. From the value $N_{|V|}$, we can conclude that another vertex $t$ is \emph{not} reachable from $s$ by checking that $N_{|V|}$ equals the number of vertices in $V-\{t\}$ reachable from $s$ within $|V|$ steps.

We adapt this practical technique in a logical setting and intend to apply it  to $\mu\snl$ in order to demonstrate that $\mathrm{DSTNCON}$ belongs to $\mu\snl$.

\begin{theorem}\label{DSTNCON-mu}
$\mathrm{DSTNCON}$ is in $\mu\snl$.
\end{theorem}

\begin{yproof}
Let us consider an arbitrary instance $x$ of the form $(G,s,t)$ given to $\mathrm{DSTNCON}$ with a directed graph $G=(\hat{V},\hat{E})$ and two vertices $s,t\in \hat{V}$. Recall that $x$ is in $\mathrm{DSTNCON}$ iff there is no path from $s$ to $t$ in $G$. For simplicity, we assume that $\hat{V}= [0,n]_{\integer}$ with $s=0$ and $t=n$ and that there is no self-loop (i.e., $(v,v)\notin \hat{E}$).
we prepare a predicate symbol $E$ so that $E(u,w)$ expresses the existence of an edge $(u,w)$ in $G$. As mentioned in Section \ref{sec:syntactic-NL}, we write $\models \psi$ to mean that a logical formula $\psi$ is true on appropriately chosen relational and domain structures $\SSS_x$ and $\DD_x$.

To encode a pair $(e,i)$ of numbers into a single number $w$, we use the formula $w=e(n+1)+i$ and we abbreviate this formula as $Enc_1(w,e,i)$. Notice that $0\leq w\leq (n+1)^2$.
In addition, we write $Enc_2(w,u,e,i)$ as a shorthand for $w=u(n+1)^2+i(n+1)+e$. It follows that $0\leq w\leq (n+1)^3$.

We then  introduce a second-order functional variable $P$ and
abbreviate as $\tilde{P}(w,e,i,u)$ the formula $B_1(w,e,i)\wedge P(w,u)$, where $B_1(w,e,i)$ expresses $0\leq w<(n+1)^2 \wedge Enc_1(w,e,i)$. Intuitively, this formula $\tilde{P}(w_1,e,i,u)$ means that (i) $w_1$ encodes $(e,i)$, (ii) there exists a path of length $i$ from $s$ ($=0$) to the vertex $u$ in $G$, and (iii) if the path reaches $e$, then it stays on $e$.
Concerning $\tilde{P}$, we demand that the following formula $\Psi_0(P)\equiv (\forall w,u,u',e,i)[\Phi_1\wedge \Phi_2\wedge \Phi_3\wedge \Phi_4]$ should be true, where
the formulas $\Phi_1$, $\Phi_2$, $\Phi_3$, and $\Phi_4$ are defined as follows.
\renewcommand{\labelitemi}{$\circ$}
\begin{itemize}\vs{-1}
  \setlength{\topsep}{-2mm}%
  \setlength{\itemsep}{1mm}%
  \setlength{\parskip}{0cm}%

\item[] $\Phi_1\equiv \neg E(u,u) \wedge [ \tilde{P}(w,e,0,u)\to e=u=0]$.

\item[] $\Phi_2\equiv Enc_1(w,0,i)\to \tilde{P}(w,0,i,0)$.

\item[] $\Phi_3\equiv \tilde{P}(w,e,i,u)\wedge \tilde{P}(w+1,e,i+1,u')\wedge u\neq e\wedge u\neq u' \to E(u,u')$.

\item[] $\Phi_4\equiv  \tilde{P}(w,e,i,e)\wedge \tilde{P}(w+1,e,i+1,u) \to u=e$.
\end{itemize}\vs{-1}
It is important to note that $\Psi_0(P)$ satisfies the second-order variable requirements because $\Phi_2$ contains only $P(w,u)$ and $P(w+1,u')$,  and $\Phi_3$ as well as $\Phi_4$ contains only $P(w,e)$ and $P(w+1,u)$.

We further introduce another formula $\Psi_1$ defined as
\[
\Psi_1(P)\equiv (\forall w) [ Enc_1(w,t,n) \to \neg \tilde{P}(w,t,n,t) ],
\]
which asserts that $t$ ($=n$) is not reachable from $s$ within $n$ steps.

For convenience, let $B_2(w,u,e,i)\equiv 0\leq w<(n+1)^3\wedge Enc_2(w,u,e,i)$. We introduce another second-order functional variable $N$ and then define $\tilde{N}(w,u,i,e,h)$ to be $B_2(w,u,e,i) \wedge N(w,h)$, which is intended to assert the existence of exactly $h$ vertices in $[0,e]_{\integer}$ reachable from $s$ by at most $i$ edges.
Let us define $\Psi_2(P,N)$ to be $(\forall w,w',u,e,i)[ \xi_1\wedge \xi_2\wedge \xi_3]$, where $\xi_1$, $\xi_2$, and $\xi_3$ are the following formulas.
\renewcommand{\labelitemi}{$\circ$}
\begin{itemize}\vs{-1}
  \setlength{\topsep}{-2mm}%
  \setlength{\itemsep}{1mm}%
  \setlength{\parskip}{0cm}%

\item[] $\xi_1 \equiv  [Enc_2(w,u,0,i) \to \tilde{N}(w,u,0,i,1)] \wedge [Enc_2(w',u,e,0) \to \tilde{N}(w',u,e,0,1)]$.

\item[]  $\xi_2 \equiv  B_2(w,u,e,i) \wedge \tilde{P}(w',e+1,i,e+1)  \to \tilde{N}(w+1,u,e+1,i, \mu h.N(w,h)+1)$.

\item[] $\xi_3 \equiv B_2(w,u,e,i) \wedge  Enc_1(w',e+1,i) \wedge \neg\tilde{P}(w',e+1,i,e+1) \to \tilde{N}(w+1,u,e+1,i, \mu h.N(w,h))$.
\end{itemize}\vs{-1}
Here, the variable $u$ is used as a ``dummy'' variable for technical reason. Note that $\xi_2$ uses $N(w,h)$, $N(w+1,h+1)$, and $P(w',e+1)$, and $\xi_3$ uses $N(w,h)$, $N(w+1,h)$, and $P(w',e+1)$. Hence, the $\mu$-term requirements (i)--(iii) are all  satisfied.
We then claim the following.

\begin{claim}\label{Psi-2-case}
Assuming that $\Psi_2(P,N)$ is true, it follows  that $\tilde{N}(w,u,e,i,h)$ is true iff $Enc_2(w,u,e,i)$ is true and  $h$ equals $| \{v\in [0,e]_{\integer} :\:  \models (\exists w') \tilde{P}(w',v,i,v) \}|$.
\end{claim}

\begin{yproof}
We show the claim for all $i$ and $e$ by induction on $i$. For simplicity, we write $W_{e,i}$ to denote the set $\{v\in [0,e]_{\integer}\mid (\exists w)[ \tilde{P}(w,v,i,v)]\mbox{ is true} \}$. Consider the case of $i=0$. Clearly, $\xi_1$ implies $\models \tilde{N}(w',u,e,0,1)$ for any $(w',u,e)$ satisfying $\models Enc_2(w',u,e,0)$.
Since $P$ behaves as a function, by $\Phi_1$, $\tilde{P}(w',e,0,0)$ is also true. We then obtain $|W_{e,0}|=1$. Assuming $\models Enc_2(w,0,e)$, it follows that $\models \tilde{N}(w,u,e,0,h)$ iff $h=1$.

Next, we assume that $i>0$. In this case, we wish to show the claim for all $e$ by induction on $e$. Assuming  $\models Enc_2(w,i,0)$, it follows by $\xi_1$ that $\tilde{N}(w,u,0,i,1)$ is true. Since $\tilde{P}(w,0,i,0)$ by $\Phi_2$, $|W_{0,i}|=1$ follows.
Next, assume that $\tilde{N}(w+1,u,e+1,i,h')$ is true. By induction hypothesis, $\models \tilde{N}(w,u,e,i,h)$ iff $\models Enc_2(w,u,e,i)$ and $h=|W_{e,i}|$.
For any $w'$ that makes $Enc_1(w',e+1,i)$ true,
if $\models \tilde{P}(w',e+1,i,e+1)$, then $\xi_2$ implies  $\models \tilde{N}(w+1,u,e+1,i,h+1)$. Since $N$ behaves as a function,  $h'$ must be $h+1$. Therefore, $h'$ matches $|W_{e+1,i}|$. By contrast, if $\models \neg\tilde{P}(w',e+1,i,e+1)$, then $\xi_3$ implies that $\tilde{N}(w+1,u,e+1,i,h)$ is true. Moreover, by definition, $W_{e,i}=W_{e+1,i}$ follows. From this equality, we obtain $h=|W_{e+1,i}|$. Therefore, we conclude that $\models \tilde{N}(w,u,e+1,i,h)$ iff $\models Enc_2(w,u,e+1,i)$ and $h=|W_{e+1,i}|$.
\end{yproof}

Another second-order functional variable $C$ is introduced. We further set  $\tilde{C}(w,u,e,i,h)\equiv B_2(w,u,e,i) \wedge C(w,h)$, which is supposed to assert that $w$ encodes $(u,e,i)$ and that $h$ equals  the total number of vertices in $[0,e]_{\integer}-\{u\}$ reachable from $s$ through at most $i$ edges.
We then define $\Psi_3(P,C) \equiv (\forall w,w',u,e,i,h) [\eta_1\wedge \eta_2\wedge \eta_3 \wedge \eta_4]$, where
\renewcommand{\labelitemi}{$\circ$}
\begin{itemize}\vs{-1}
  \setlength{\topsep}{-2mm}%
  \setlength{\itemsep}{1mm}%
  \setlength{\parskip}{0cm}%

\item[] $\eta_1 \equiv  [Enc_2(w,0,0,i) \to \tilde{C}(w,0,0,i,0)] \wedge [u\geq1 \wedge Enc_2(w',u,e,0) \to \tilde{C}(w',u,e,0,1)]$.

\item[] $\eta_2 \equiv  \tilde{C}(w,u,e,i,h) \wedge u = e+1 \to \tilde{C}(w+1,u,e+1,i,h)$.

\item[] $\eta_3 \equiv  B_2(w,u,e,i) \wedge Enc_1(w',e+1,i) \wedge \neg \tilde{P}(w',e+1,i,e+1)  \wedge u\neq e+1 \to \tilde{C}(w+1,u,e+1,i, \mu h.C(w,h))$.

\item[] $\eta_4 \equiv  Enc_1(w',u,i+1) \wedge \neg \tilde{P}(w',u,i+1,u) \wedge B_2(w,u,e,i) \wedge E(e+1,u)   \to  \tilde{C}(w+1,u,e+1,i, \mu h.C(w,h))$.

\item[] $\eta_5 \equiv  B_3(w,u,e,i) \wedge \tilde{P}(w',e+1,i,e+1) \wedge  u\neq e+1 \to  \tilde{C}(w+1,u,e+1,i, \mu h. C(w,h)+1)$.
\end{itemize}

Note that $\eta_3$ as well as $\eta_4$ and $\eta_5$ contains $C(w,h)$, $C(w+1,h)$, and $P(w',e+1)$. This makes the $\mu$-term requirements satisfied.

\begin{claim}\label{Psi-3-case}
Consider the case where $\Psi_3(P,C)$ is true. Assuming that $Enc_1(w',u,i+1) \wedge \neg \tilde{P}(w',u,i+1,u)$ is true, it follows that $\tilde{C}(w,u,e,i,h)$ is true iff $Enc_2(w,u,e,i)$ is true and $h$ equals $|\{ v\in [0,e]_{\integer}-\{u\}:\: \models (\exists w')[\tilde{P}(w',v,i,v)\wedge \neg E(v,u)] \}|$.
\end{claim}

\begin{yproof}
We show the claim for all $i$, $e$, and $u$ by induction on $i$.
We write $V_{u,e,i}$ for the set $\{ v\in[0,e]_{\integer}-\{u\} :\: \models (\exists w)[\tilde{P}(w,v,i,v)\wedge \neg E(v,u)] \}$.
Let us consider the base case of $i=0$.
For any $(w,u,e)$ that makes $Enc_2(w,u,e,0)$ true, it follows
by $\eta_1$ that $\models \tilde{C}(w,0,e,0,0)$ and $\models \tilde{C}(w,u,e,0,1)$ when $u\geq1$.
From $\Phi_1$, we also obtain $V_{0,e,0}=\setempty$ and $V_{u,e,0}=\{0\}$  whenever $u\geq1$.

Next, let us consider the inductive case of $i+1$. We study each case of $e$ inductively. When $e=0$, by $\eta_1$, if $u=0$, then $\tilde{C}(q,u,0,i+1,0)$ is true. Moreover, by $\Phi_1$, we obtain $V_{u,0,i+1}=\setempty$. In contrast, when $u\geq1$, we obtain $\models \tilde{C}(w,u,0,i+1,1)$ and $V_{u,0,i+1}=\{0\}$.

For the case of $e+1$, induction hypothesis implies that $\models \tilde{C}(w,u,e,i+1,h)$ iff $\models Enc_2(w,u,e,i+1)$ and $h=|V_{u,e,i+1}|$. Let us assume that $Enc_1(w',u,i+1)\wedge \neg \tilde{P}(w',u,i+2,u)$ is true. Assuming $\models \tilde{C}(w,u,e,i+1,h)$, if $E(e+1,u)$ is true, then $\eta_4$ makes $\tilde{C}(w+1,u,e+1,i+1,h)$ true. Moreover, we obtain $e+1\notin V_{u,e+1,i+1}$, and thus $|V_{u,e+1,i+1}|=|V_{u,e,i+1}|=h$ follows.
On the contrary, assume that $\models \neg E(e+1,u)$. If $u=e+1$, then $\eta_2$ implies $\models \tilde{C}(w+1,u,e+1,i+1,h)$. Since $e+1\notin V_{u,e+1,i+1}$, we conclude that $|V_{u,e+1,i+1}|=|V_{u,e,i+1}|=h$.
Next, we assume that $u\neq e+1$. If $\tilde{P}(w'',e+1,i+1,e+1)$ is true, then $\eta_5$ derives $\models \tilde{C}(w+1,u,e+1,i+1,h+1)$ and we obtain $e+1\in V_{u,e+1,i+1}$.
This last result leads to the equality of  $V_{u,e+1,i+1}=V_{u,e,i+1}\cup\{e+1\}$, and thus $|V_{u,e+1,i+1}|=|V_{u,e,i+1}|+1 = h+1$ follows.
By contrast, if $\neg \tilde{P}(w'',e+1,i+1,e+1)$ is true, then $\eta_3$ yields $\models \tilde{C}(w+1,u,e+1,i+1,h)$. Moreover, we obtain $e+1\notin V_{u,e+1,i+1}$ by the definition of $V_{u,e+1,i+1}$.
This consequence further implies that $|V_{u,e+1,i+1}| = |V_{u,e,i+1}|=h$. Therefore, we conclude that $\models \tilde{C}(w,u,e+1,i+1,h)$ iff $\models Enc_2(w,u,e+1,i+1)$ and $h=|V_{u,e+1,i+1}|$.
\end{yproof}

Moreover, we introduce $\Psi_4(P,N,C)$ defined as
\renewcommand{\labelitemi}{$\circ$}
\begin{itemize}\vs{-1}
\item[] $\Psi_4\equiv  B_1(w,u,i)\wedge B_2(w',u,n,i) \to [  \neg\tilde{P}(w,u,i,u) \leftrightarrow  \tilde{C}(w',u,n,i, \mu h.N(w',h) )$.
\end{itemize}\vs{-1}
Notice that the $\mu$-term requirements are clearly satisfied.

The desired sentence $\Phi$ is finally set to be
\[
(\exists^f P,N,C) [ \Psi_0(P) \wedge \Psi_1(P) \wedge
%(\forall e,i) [\tilde{P}(w,e,i,e)\to T^{(-)}(e,i)] \wedge
\Psi_2(P,N) \wedge \Psi_3(P,C) \wedge \Psi_4(P,N,C)]
\]
with the above three additional formulas $\Psi_2$, $\Psi_3$, and $\Psi_4$ to $\Psi_0$ and $\Psi_1$.
By induction on $i$, we verify the following claim.

\begin{claim}\label{claim-3}
Assume that $\Phi$ is true. For all $(w,u,i)$ satisfying $Enc_1(w,u,i)$, it follows that  $\tilde{P}(w,u,i,u)$ is true iff  $s\stackrel{\leq i}{\leadsto} u$ holds.
\end{claim}

\begin{yproof}
Let us recall the notations $W_{e,i}$ and $V_{u,e,i}$ from the proofs of Claims \ref{Psi-2-case} and \ref{Psi-3-case}.
We proceed the proof by induction on $i$. When $i=0$, it clearly follows that $\models \tilde{P}(w,u,0,u)$  iff $u=0$ iff $s\stackrel{\leq 0}{\leadsto}u$.

In what follows, we examine the case of $i+1$.
Let us assume that $s\stackrel{\leq i+1}{\leadsto} u$. If $s\stackrel{\leq i}{\leadsto} u$, then we simply apply induction hypothesis. Thus, we now assume that $s\stackrel{\leq i}{\leadsto} u$ does not hold.
Toward a contradiction, we further assume that
$\tilde{P}(w,u,i+1,u)$ is false; that is, $\neg \tilde{P}(w,u,i+1,u)$ is true. By induction hypothesis, the claim implies that $W_{n,i} =\{v\in[0,n]_{\integer}\mid  s\stackrel{\leq i}{\leadsto}v\}$ and $V_{u,n,i} =\{v\in[0,n]_{\integer}-\{u\}\mid  s\stackrel{\leq i}{\leadsto}v \text{ and } \models \neg E(v,u)\}$.
Since $s\stackrel{\leq i+1}{\leadsto} u$, there exists a vertex $x$ for which $s\stackrel{\leq i}{\leadsto} x$ and $\models E(x,u)$. From this follows $x\neq u$.
By definition, we obtain $x\in W_{n,i}$ but $x\notin V_{u,n,i}$.  Therefore, $\models \tilde{C}(w',u,n,i+1,h) \wedge \tilde{N}(w',u,n,i+1,h')$ holds for $h=|V_{u,n,i+1}|$ and $h'=|W_{n,i}|$. Since $h\neq h'$, $\Psi_4$ implies that $\tilde{P}(w,u,i+1,u)$ is true. This is a contradiction.

We then show the converse. We first prove by induction on $i$ that (*) for any $x$, if $\models \tilde{P}(w,u,i,x)$, then $s\stackrel{\leq i}{\leadsto} x$ holds. When $i=0$, if $\models \tilde{P}(w,u,0,x)$, then $\Phi_1$ leads to $u=x=0$.
In addition, we obtain $s\stackrel{\leq 0}{\leadsto} s$. Next, we consider the case of $i+1$. Assume that $\tilde{P}(w,u,i+1,x)$ is true. If there is an element $y$ satisfying $\models \tilde{P}(w',u,i,y)\wedge E(y,x)$, then induction hypothesis leads to both $s\stackrel{\leq i}{\leadsto} y$ and $\models E(y,x)$, which yield $s\stackrel{\leq i+1}{\leadsto} x$. Conversely, let us assume that    $\tilde{P}(w',u,i,y)\wedge E(y,x)$ is false. If $\models \tilde{P}(w',u,i,y)$, then our assumption yields $\models \neg E(y,x)$. However, since $\tilde{P}(w,u,i+1,x)$ is true, $\Phi_3$ implies $\models E(y,x)$. This is a contradiction. Hence, Statement (*) is true. As a special case of (*), by setting $x=u$, we conclude that $\models \tilde{P}(w,u,i,u)$ implies $s\stackrel{\leq i}{\leadsto} u$.
\end{yproof}

Assuming $\models Enc_1(w,t,n)$, $\Psi_1(P)$ implies that  $\neg\tilde{P}(w,t,n,t)$ is true. This means that, by Claim \ref{claim-3},  there is no path from $s$ to $t$ in $G$. It then follows that $\Phi$ is true iff $G$ has no path from $s$ to $t$. Hence, $\Phi$ syntactically expresses $\mathrm{DSTNCON}$. This concludes that $\mathrm{DSTNCON}$ is in $\mu\snl$.
\end{yproof}

%%%%%%%
%%%%%%%
\section{Monotone Variant of SNL}\label{sec:two-variants}

We have discussed the basic structural properties in Section \ref{sec:structure}. We next intend to expand the scope of our study on $\snl$ in hopes of making its direct application to other areas of computer science.
In the past literature, there have been intensive studies on a wide range of variations of $\snp$.
We focus on the monotone restriction of $\snp$ ($\mathrm{MonoSNP}$). In natural analogy, we intend to investigate similar concepts induced from $\snl$ and to study their specific characteristics.

%%%%%%
\subsection{Monotone SNL (or MonoSNL)}\label{sec:monotone}

In the polynomial-time setting, Feder and Vardi \cite{FV93,FV99} studied structural properties of two restricted versions of $\mathrm{SNP}$, called \emph{monotone SNP} (MonoSNP) and \emph{monotone monadic SNP}  (MMSNP). Bodinsky, Chen, and Feder \cite{BCF12} later gave a characterization of MMSNP under a certain natural restriction.
In a similar fashion, let us consider a natural subclass of $\snl$, which we call the \emph{monotone SNL} or succinctly \emph{MonoSNL}.
Suppose that our vocabulary $\VV$ contains predicate symbols $S_1,S_2,\ldots,S_d$.
Given a formula $\Phi$ over $\VV$, we then transform it to its conjunctive normal form (CNF).
This formula $\Phi$ is said to be \emph{monotone} if the predicates $S_j$'s appearing in this CNF formula are all negative (i.e., of the form $\neg S_i(v_{i_1},\ldots,v_{i_k})$).
In what follows, we wish to study the expressibility of monotone $\snl$ sentences.

\begin{definition}\label{def-monoSNL}
The notation $\mathrm{Mono}\snl$ denotes the class of all decision problems that are syntactically expressed by monotone $\snl$ sentences.
\end{definition}

It turns out that $\mathrm{MonoSNL}$ contains natural $\nl$ problems. For example, the problem $\mathrm{2COLOR}$, discussed in Example \ref{2COLOR}, belongs to $\mathrm{MonoSNL}$. Another example is the problem $\mathrm{exact}3\dstcon$.

\begin{yexample}\label{eaxct3DSTCON-monotone}
{\rm The problem $\mathrm{exact}3\dstcon$ is in $\mathrm{MonoSNL}$. Consider any instance $(G,s,t)$ of $\mathrm{exact3}\dstcon$ with $G=(V_G,E_G)$ and $s,t\in V_G$. For simplicity, we assume that $s$ has indegree $0$. We naturally extend $G$ by including the edge $(t,t)$ and write $G^{(ext)}$ for this extended graph.
We then introduce a predicate symbol $E$, which represents the edge set  $E_G$.
For a second-order variable $P$, $P(i,u)$ semantically indicates that a given vertex $u$ is the $i$th element of a path of $G^{(ext)}$.
We define an $\snl$ sentence $\Phi$ to be $(\exists^fP) (\forall i,u,v,v_1,v_2,v_3)[P(0,s)\wedge P(n,t)\wedge \Phi_1\wedge \Phi_2\wedge \Phi_3]$, where
$\Phi_1\equiv P(0,s)\wedge P(1,v)\wedge (\bigwedge_{k=1}^{3} E(s,v_k)) \wedge v_1\neq v_2\neq v_3\neq v_1 \to \bigvee_{k=1}^{3}(v=v_k)$,
$\Phi_2\equiv 0<i<n\wedge P(i,u)\wedge P(i+1,v)\wedge (\bigwedge_{k=1}^{2} E(u,v_k)) \wedge v_1\neq v_2 \to \bigvee_{k=1}^{2}(v=v_k)$, and  $\Phi_3\equiv 0\leq i<n \wedge P(i,t)\to P(i+1,t)$.
Intuitively, $\Phi_1$ means that the path contains an edge from $s$, $\Phi_2$ means that, if vertex $u$ is in the path, then the path contains an edge from $u$, and $\Phi_3$ means that, if the path reaches $t$ at some point, then the path stays on $t$.
It thus follows by definition that $\Phi$ syntactically expresses $\mathrm{exact}3\dstcon$.
The formula $\Phi$ is monotone because $\Phi_1$ is rewritten as  $\neg P(0,s)\vee  \neg P(1,v)\vee  (\bigvee_{k=1}^{3} \neg E(s,v_k)) \vee \neg(v_1\neq v_2\neq v_3\neq v_1) \vee \bigvee_{k=1}^{3}(v=v_k)$ and $\Phi_2$ is rewritten as
$\neg(0<i<n)\vee  \neg P(i,u)\vee \neg P(i+1,v)\vee (\bigvee_{k=1}^{2} \neg E(u,v_k)) \vee v_1= v_2 \vee  \bigvee_{k=1}^{2}(v=v_k)$.
}
\end{yexample}

Example \ref{eaxct3DSTCON-monotone} immediately leads to the following consequence.

\begin{corollary}
$\nl= \; \Lreduces\!(\mathrm{Mono}\snl)$.
\end{corollary}

\begin{yproof}
Obviously, $\mathrm{MonoSNL}\subseteq \snl\subseteq \nl$ follows by definition. We thus obtain $\Lreduces\!(\mathrm{MonoSNL}) \subseteq \; \Lreduces\!(\nl) = \nl$.
By Example \ref{eaxct3DSTCON-monotone}, $\mathrm{MonoSNL}$ contains the decision problem $\mathrm{exact}3\dstcon$.  Since $\mathrm{exact}3\dstcon$ is complete for $\nl$ under $\dl$-m-reductions (see, e.g., \cite{Yam17a}), we instantly obtain $\nl\subseteq \; \Lreduces\!(\mathrm{MonoSNL})$.
\end{yproof}

Next, we ask whether the \emph{dichotomy theorem} holds for $\mathrm{MonoSNL}$; namely,   every decision problem in $\mathrm{MonoSNL}$ is either in $\dl$ or complete for $\nl$ under $\dl$-m-reductions. We do not know that this is the case. This situation is compared to the case of $\snl$ as shown in Corollary \ref{dichotomy-L-NL}. However, if $\mathrm{MonoSNL}$ equals $\snl$, then $\mathrm{MonoSNL}$ is unlikely enjoy the dichotomy theorem.

%%%%

In close relation to the dichotomy theorem, as a simple example, we examine the ``monotone'' segments of 2SAT ($Polar^{(+)}\mbox{-}\twosat$ and $Polar^{(-)}\mbox{-}\twosat$) founded on the notion of the \emph{polarity} of 2CNF Boolean formulas.
Let us recall that a 2CNF formula $\psi$ has the form $\bigwedge_{i=1}^{k}\phi_i$ with $\phi_i\equiv z_{i1}\vee z_{i2}$, where each $z_{ij}$ is a literal. If every clause $\phi_i$ is of the form either $x\vee y$ or $\overline{x}\vee \overline{y}$ for variables $x$ and $y$, then $\psi$ is said to have \emph{positive polarity}.
In contrast, if every $\phi_i$ has the form $\overline{x}\vee y$ (or $x\vee \overline{y}$), then $\psi$ has \emph{negative polarity}. The problem $Polar^{(+)}\mbox{-}\twosat$ (resp.,  $Polar^{(-)}\mbox{-}\twosat$) is then defined as the collection of all satisfiable 2CNF formulas that have positive (resp., negative) polarity. These problems $Polar^{(+)}\mbox{-}\twosat$ and $Polar^{(-)}\mbox{-}\twosat$ can be syntactically expressed by monotone $\snl$ sentences.

\begin{lemma}\label{polar-MonoSNL}
$Polar^{(+)}\mbox{-}\twosat$ and $Polar^{(-)}\mbox{-}\twosat$ are both in $\mathrm{MonoSNL}$.
\end{lemma}

\begin{yproof}
We first consider the case of $Polar^{(+)}\mbox{-}\twosat$.
Let $\psi$ denote any instance of the form $\bigvee_{i=1}^{t}\phi_i$ given to the decision problem $Polar^{(+)}\mbox{-}\twosat$, where each clause $\phi_i$ is either $x\vee y$ or $\overline{x}\vee\overline{y}$ for appropriate variables $x$ and $y$.
Let $V=\{x_1,x_2,\ldots,x_n\}$ denote the set of all variables in $\psi$ and write $\overline{V}$ for the set $\{\overline{x_1},\overline{x_2},\ldots,\overline{x_n}\}$ of negated variables.
To simplify a later argument, we write $z_1,z_2\ldots,z_n$ to denote $x_1,\ldots,x_n$ and $z_{n+1},\ldots,z_{2n}$ for $\overline{x_1},\ldots,\overline{x_n}$, respectively.

We prepare a predicate symbol $C$ and a second-order functional variable $T$. Let $C(i,j)$ express that a pair $(z_i,z_j)$ appears as a clause of $\psi$ in the form $z_i\vee z_j$ or $z_j\vee z_i$. Moreover, $T$ behaves as a truth assignment so that $T(i,1)$ (resp., $T(i,0)$) means that literal $z_i$ is assigned to be true (resp., false).
We set $\Phi\equiv (\exists^{f}T )(\forall i, u, i', j', i'', j'') [\Phi_1(T,i,u) \wedge \Phi_2(C,T,i',j') \wedge \Phi_3(C,i'',j'')]$, where $\Phi_1\equiv  (T(i,0) \wedge T(n+i,1))\vee (T(i,1)\wedge T(n+i,0))$, $\Phi_2\equiv C(i',j') \rightarrow T(i',1)\vee T(j',1)$, and $\Phi_3\equiv C(i'',j'') \rightarrow
(1\leq i''\leq n \wedge 1\leq j''\leq n) \vee (n+1\leq i'' \leq 2n\wedge n+1\leq j''\leq 2n)$.
Notice that $\Phi_2$ is logically equivalent to $\neg C(i',j')\vee T(i',1)\vee T(j',1)$ and that $\Phi_3$ is equivalent to $\neg C(i'',j'')\vee (1\leq i'' \leq n \wedge 1\leq j''\leq n) \vee (n+1\leq i''\leq 2n \wedge n+1\leq j''\leq 2n)$. Therefore, $\Phi$ is monotone.
It is not difficult to see that $\psi$ is satisfiable iff there is a domain structure that satisfies $\Psi$.

For the case of $Polar^{(-)}\mbox{-}\twosat$, we change the above defined formula $\Phi_3$ to $\Phi'_3$ of the form $C(i'',j'') \rightarrow (1\leq i'' \leq n  \wedge n+1\leq j''\leq 2n) \vee (1\leq j'' \leq n \wedge  n+1\leq i''\leq 2n)$. An argument similar to the case of $Polar^{(+)}\mbox{-}\twosat$ also works for $Polar^{(-)}\mbox{-}\twosat$.
\end{yproof}

In the log-space setting, by contrast, positive and negative polarities act quite differently.

\begin{proposition}
$Polar^{(-)}\mbox{-}\twosat$ is in $\dl$ and $Polar^{(+)}\mbox{-}\twosat$ is complete for $\nl$ under $\dl$-m-reductions.
\end{proposition}

\begin{yproof}
Notice by Lemma \ref{polar-MonoSNL} that $Polar^{(+)}\mbox{-}\twosat\in\nl$. To prove that $\dl$-m-hardness of $Polar^{(+)}\mbox{-}\twosat$ for $\nl$, we wish to reduce $\twosat$ to $Polar^{(+)}\mbox{-}\twosat$ by $\dl$-m-reductions.
Let $\psi$ be any instance of the form $\bigwedge_{i=1}^{k}\phi_i$ given to the decision problem $\twosat$ with $\phi_i\equiv z_{i1}\vee z_{i2}$, where each $z_{ij}$ is a literal. For each clause $\phi_i$, if it has the form of either $x\vee y$ or $\bar{x}\vee \bar{y}$, then we set $\tilde{\phi}_i$ to be $\phi_i$. When $\phi_i$ is of the form $x\vee \bar{y}$, we prepare a new variable $z$ and define $\tilde{\phi}_i^{(1)}\equiv x\vee z$ and $\tilde{\phi}_i^{(2)}\equiv \bar{z}\vee \bar{y}$. We then set  $\tilde{\phi}\equiv \tilde{\phi}_i^{(1)}\wedge \tilde{\phi}_i^{(2)}$. The case of $\bar{x}\vee y$ is similarly treated. Notice that $\phi_i$ is satisfiable exactly when so is $\tilde{\phi}_i$.
Since the formula $\bigwedge_{i=1}^{k}\tilde{\phi}_i$ has positive polarity, it follows that $\phi\in \twosat$ iff $\tilde{\phi}\in Polar^{(+)}\mbox{-}\twosat$. This means that $\twosat \Lreduces Polar^{(+)}\mbox{-}\twosat$.

Next, we wish to solve $Polar^{(-)}\mbox{-}\twosat$. If we replace $\bar{x}\vee y$ to $x\rightarrow y$ and $x\vee \bar{y}$ to $y\rightarrow x$, then we obtain a directed graph whose vertices are labeled by variables. In this case, if we assign $1$ (true) to all variables, then we can make the instance formula satisfiable. Thus, $Polar^{(+)}\mbox{-}\twosat$ falls in $\dl$.
\end{yproof}

%%%

Now, we look into $\mathrm{SNL}_{\omega}$ and its monotone version, $\mathrm{MonoSNL}_{\omega}$. These complexity classes contain quite natural restrictions of $\mathrm{SNL}$ and $\mathrm{MonoSNL}$.
To see this fact, let us first recall that any second-order functional variable, say, $P$ used in an $\snl$ sentence acts as a function mapping natural numbers to certain ``objects'' specified by an underlying domain structure $\DD$.
Here, we consider the special case where these objects are taken from the binary set $\{0,1\}$.
In other words, all second-order functional variables represent functions from natural numbers to $\{0,1\}$.
This makes $P$ behave like a single argument predicate by interpreting $P(\cdot,1)$ (resp., $P(\cdot,0)$) as ``true'' (resp., ``false''). We call any SNL sentence with this restriction a \emph{binary SNL sentence}.

\begin{definition}\label{def-BSNL}
The notation $\mathrm{BSNL}$ expresses the subclass of $\snl$ characterized by binary SNL sentences. With the use of $\mathrm{BSNL}$  in place of $\snl$, we define $\mathrm{MonoBSNL}$ from $\mathrm{MonoSNL}$.
\end{definition}

\begin{lemma}
$\mathrm{BSNL} \subseteq \snl_{\omega}$ and $\mathrm{MonoBSNL} \subseteq \mathrm{MonoSNL}_{\omega}$.
\end{lemma}

\begin{yproof}
Let $L$ denote any language in $\mathrm{BSNL}$ and take a binary $\snl$ sentence $\Phi$ that syntactically expresses $L$. This $\Phi$ has the form $\exists^fP_1\cdots\exists^fP_l \forall\boldvec{i}\forall\boldvec{j}[  \psi(P_1,\ldots,P_l,\boldvec{i},\boldvec{j}) ]$. Notice that each  second-order functional variable $P_h$ represents a function mapping natural numbers to $\{0,1\}$. Now, we consider the logically equivalent sentence $\tilde{\Phi}\equiv \exists^fP_1\cdots \exists^f P_l\forall\boldvec{i}\forall\boldvec{j} [ \psi(P_1,\ldots,P_l,\boldvec{i},\boldvec{j}) \wedge (\bigwedge_{h=1}^{l} Func(P_h))]$.
Note that the sentence $Func(P_h)$ is logically expressed as $(\forall a)[(P_h(a,0)\vee P_h(a,1))\wedge (\neg P_h(a,0)\vee \neg P_h(a,1))]$.
Hence, $\tilde{\Phi}$ can be rewritten as $\exists^fP_1\cdots \exists^fP_l\forall\boldvec{i}\forall\boldvec{j} \forall a_1\cdots \forall a_l [ \psi(P_1,\ldots,P_l,\boldvec{i},\boldvec{j}) \wedge (\bigwedge_{h=1}^{l} \xi_h(a_h))]$, where $\xi_h(a_h) \equiv (P_h(a_h,0)\vee P_h(a_h,1))\wedge (\neg P_h(a_h,0)\vee \neg P_h(a_h,1))$. This formula  satisfies the second-order variable requirements. We thus obtain $\mathrm{BSNL}\subseteq \snl_{\omega}$.

The last part of the lemma follows similarly.
\end{yproof}

%%%%
\subsection{Relationships to Constraint Satisfaction Problems}

Feder and Vardi \cite{FV99} demonstrated that every problem in MMSNP is polynomial-time equivalent to a constraint satisfaction problem (CSP).
We fix a set $V$ fo ``objects'' and a set $\Gamma$ of constraint functions $f$ mapping $V^k$ to $\{0,1\}$, where $k$ is the arity of $f$.
A \emph{CSP over $(V,\Gamma)$} consists of the following items:  a set $X=\{x_1,x_2,\ldots,x_n\}$ of variables and a set $C$ of constraints of the form $(f,(x_{i_1},x_{i_2},\ldots,x_{i_k}))$ with $f\in \Gamma$ and $x_{i_1},\ldots,x_{i_k}\in X$.
An \emph{assignment} $\rho$ is a function from $X$ to $V$. Given such an assignment $\rho$, we evaluate each constraint $(f,(x_{i_1},x_{i_2},\ldots,x_{i_k}))$ by computing the Boolean value  $f(\rho(x_{i_1}),\rho(x_{i_2}),\ldots,\rho(x_{i_k}))$.
A CSP $(X,C)$ is \emph{satisfiable} if there exists an assignment $\rho$ that makes all constraints of the CSP satisfied.
A CSP is said to be \emph{of arity at most $d$} if all of its constraints have arity at most $d$.
Given a pair $(V,\Gamma)$, we write $CSP_d(V,\Gamma)$ for the collection of all satisfiable CSPs over $(V,\Gamma)$ of arity at most $d$.
When $V=\{0,1\}$, in particular, a CSP over $(V,\Gamma)$ is called a \emph{binary CSP over $\Gamma$}.
To express the subproblem of $CSP_d(\Gamma)$ restricted to only binary CSPs over $\Gamma$, we use the special notation of $BCSP_d(\Gamma)$.

Let $\Gamma_{OR}$ denote the set of four constraint functions $f$ defined by setting  $f(x,y)$ to be one of $x\vee y$, $\bar{x}\vee y$, $x\vee \bar{y}$, and $\bar{x}\vee \bar{y}$ for two variables $x$ and $y$.
We first demonstrate that
$\mathrm{BCSP}_2(\Gamma_{OR})$ characterizes $\nl$.

\begin{proposition}\label{NL-vs-BCSP}
$\nl =\; \Lreduces\!(BCSP_2(\Gamma_{OR}))$.
\end{proposition}

Before proving this proposition, we claim the following basic properties.

\begin{lemma}\label{CSP-monoSNL}
Assume that $V$ and $\Gamma$ are finite sets. All CSPs over $(V,\Gamma)$ of arity at most $2$ belong to $\mathrm{MonoSNL}$ and all binary CSPs over $\Gamma$ of arity at most $2$ belong to $\mathrm{MonoBSNL}$.
\end{lemma}

\begin{yproof}
Given a CSP $(X,\hat{C})$ over $(V,\Gamma)$ of arity at most $2$, let $X$ denote a set of variables $x_1,x_2,\ldots,x_n$ and, for each $d\in\{1,2\}$,  let $\hat{C}_d$ denote a collection of constraints of arity exactly $d$ in $\hat{C}$.
Clearly, $\hat{C}$ coincides with $\hat{C}_1\cup \hat{C}_2$.

For simplicity, we express each variable $x_i$ as $i$.
For each index $d\in\{1,2\}$, we introduce a predicate symbol $C_d$ so that  $\models C_d(f,i_1,\ldots,i_d)$ iff $(f,(x_{i_1},\ldots,x_{i_d}))\in \hat{C}_d$.
We also introduce two more predicate symbols $S_1$ and $S_2$ so that, for any $d\in\{1,2\}$ and for any tuple $(v_1,\ldots,v_d)\in V^d$,  $S_d(f,v_1,\ldots,v_d)$ is true iff $f(v_1,\ldots,v_d) =0$. Let us consider the following sentence: $\Phi\equiv \exists^fP\forall f \forall i \forall v \forall i_1 \forall i_2 \forall v_1\forall v_2 [\Phi_1(P,C_1,S_1,f,i,v) \wedge \Phi_2(P,C_2,S_2,f,i_1,i_2,v_1,v_2)]$,
where $\Phi_1\equiv 1\leq i\leq n \wedge C_1(f,i)\wedge P(i,v)\to \neg S_1(f,v)$
and $\Phi_2\equiv 1\leq i_1\leq n\wedge 1\leq i_2\leq n\wedge C_2(f,i_1,i_2)\wedge P(i_1,v_1)\wedge P(i_2,v_2)\to \neg S_2(f,v_1,v_2)$. Here, $P$ represents an assignment, say, $\rho$ such that $P(i,v)$ is true iff $\rho(i)=v$.

The formula $\Phi$ is obviously a monotone $\snl$ sentence because $\Phi_1$ is equivalent to $\neg(1\leq i\leq n) \vee \neg C_1(f,i)\vee \neg P(i,v) \vee \neg S_1(f,v)$ and $\Phi_2$ is equivalent to $\neg(1\leq i_1\leq n)\vee \neg(1\leq i_2\leq n) \vee \neg C_2(f,i_1,i_2)\vee \neg P(i_1,v_1)\vee \neg P(i_2,v_2) \vee \neg S_2(f,v_1,v_2)$. It also follows by the definition of $\Phi$ that $\Phi$ is true iff $(X,\hat{C})$ is satisfiable.

For the second part of the proposition, we start with a binary CSP of arity at most $2$.
In this case, we need to replace $P(i,v)$ in the above argument   by $P(i,v)\wedge 0\leq v\leq 1$. Similarly, $P(i_1,v_1)$ and $P(i_2,v_2)$ should be replaced. A similar argument as above proves the desired second part.
\end{yproof}

\begin{proofof}{Proposition \ref{NL-vs-BCSP}}
By Lemma \ref{CSP-monoSNL}, $BCSP_2(\Gamma)$ is contained in $\mathrm{MonoSNL}$ ($\subseteq \nl$) for any $\Gamma$. In particular, $BCSP_2(\Gamma_{OR})$ belongs to $\nl$.
It therefore suffices to prove that $\twosat \Lreduces BCSP_2(\Gamma_{OR})$.
Let $\phi$ denote any 2CNF formula built from a variable set $X=\{x_1,x_2,\ldots,x_n\}$ and a clause set $C\subseteq (X\cup \bar{X})^2$, where $\bar{X}=\{\overline{x_1},\overline{x_2},\ldots,\overline{x_n}\}$.  Whenever each clause in $C$ is a ``single'' literal, say, $z$, we replace it with the clause $z\vee z$. It thus possible for us to assume that $\phi$ has the form $\bigwedge_{(z_{i_1},z_{i_2})\in C}(z_{i_1}\vee z_{i_2})$, where $z_{i_1}$ and $z_{i_2}$ are literals. Given a clause $z_{i_1}\vee z_{i_2}$, we choose $f\in\Gamma_{OR}$ such that $f(x_{i_1},x_{i_2})$ equals $z_{i_1}\vee z_{i_2}$, where $x_{i_1}$ and $x_{i_2}$ are the underlying variables of $z_{i_1}$ and $z_{i_2}$, respectively. Let $C$ denote the  collection $C$ of all such constraints $(f,(x_{i_1},x_{i_2}))$.

Let $I_{\phi}$ denote the CSP made up from $X$ and $C$.
It then follows that $\phi$ is satisfiable iff $I_{\phi}$ is satisfiable. We thus conclude that $\twosat \Lreduces BCSP_2(\Gamma_{OR})$. Since $\twosat$ is $\dl$-m-complete for $\nl$, so is $BCSP_2(\Gamma_{OR})$.
Therefore, we obtain $\nl= \;\Lreduces\!(BCSP_2(\Gamma_{OR}))$.
\end{proofof}

%%%

We now turn to the dichotomy theorem for $\mathrm{MonoBSNL}$. Recall that the dichotomy theorem for $\CC$ means that every decision problem in $\CC$ is either in $\dl$ or $\nl$-complete.
It is not yet known that $\mathrm{MonoBSNL}$ enjoys the dichotomy theorem; in sharp contrast, however, the dichotomy theorem holds for $BCSP_2(\Gamma)$.  Allender \etal~\cite{ABI+09} showed that, for any set $\Gamma$, $BCSP(\Gamma)$ is $\mathrm{AC}^0$-isomorphic either to $0\Sigma^*$ or to the ``standard'' complete problem (under $\mathrm{AC}^0$-reductions) for one of the following complexity classes: $\np$, $\p$, $\oplus\dl$, $\nl$, and $\dl$. In our restricted case, from their result, we obtain the following statement.

\begin{proposition}\label{BCSP-dichotomy}
For any set $\Gamma$, $BCSP_2(\Gamma)$ is either in $\dl$ or $\nl$-complete.
\end{proposition}

%%%

Inspired by the connection between $\mathrm{MMSNP}$ and CSPs, we wish to demonstrate a close relationship between $\mathrm{MonoBSNL}$ and $BCSP_2(\Gamma)$.

\begin{theorem}\label{MonoBSNL-to-BCSP}
For any decision problem $\Xi$ in $\mathrm{MonoBSNL}$, there exists a finite set $\Gamma$ such that $\Xi$ is $\dl$-m-reducible to $BCSP_2(\Gamma)$.
\end{theorem}

\begin{yproof}
Given a decision problem $\Xi$ in $\mathrm{MonoBSNL}$, we take a binary $\snl$ sentence $\Phi \equiv \exists^f\boldvec{P} \forall \boldvec{i} \forall \boldvec{y} [ \bigwedge_{j=1}^{t} \psi_j(\boldvec{P},\boldvec{i},\boldvec{y})]$, as stated in Definition \ref{SNL-sentence}, provided that $\boldvec{P}=(P_1,P_2,\ldots,P_l)$.
We then construct
$\bigwedge_{j=1}^{t} \bigwedge_{\underline{\boldvec{i}},\underline{\boldvec{y}}} \psi_j (\boldvec{P},\underline{\boldvec{i}},\underline{\boldvec{y}})$ by assigning all possible values $(\underline{\boldvec{i}},\underline{\boldvec{y}})$ sequentially one by one.
We wish to convert each of the obtained formulas, $\psi_j(\boldvec{P},\underline{\boldvec{i}},\underline{\boldvec{y}})$, into an equivalent CNF formula. By the second-order variable requirements (i)--(ii), each conjunct (in the obtained formula) has at most two appearances of $P_{k}(\cdot)$ or $\neg P_{k'}(\cdot)$. We then evaluate all the other predicates (except for $P_k(\cdot)$ and $\neg P_{k'}(\cdot)$) to be either $0$ (false) or $1$ (true), resulting in a formula consisting only of $P_k(\cdot)$ and $\neg P_{k'}(\cdot)$.

As a concrete example, recall $\mathrm{2COLOR}$ in Example \ref{2COLOR}. In this case, $\psi_j(\boldvec{P},\underline{\boldvec{i}},\underline{\boldvec{y}})$ is either of the following forms:
$\neg C(i,d) \vee 0\leq d\leq 1$ or $\neg E(i',d')\vee \neg C(i',d')\vee \neg C(j',e')\vee d'\neq e'$. Since $(i,d,i',d',j',e')$ are considered to be fixed in $\psi_j(\boldvec{P},\underline{\boldvec{i}},\underline{\boldvec{y}})$, the  subformulas, $0\leq d\leq 1$, $\neg E(i',d')$, and $d'\neq e'$, are all evaluated to be either $0$ or $1$. We then remove all such subformulas evaluated to be $0$ and remove the entire  $\psi_j(\boldvec{P},\underline{\boldvec{i}},\underline{\boldvec{y}})$  if its  evaluation is $1$. We then rephrase $\psi_j(\boldvec{P},\underline{\boldvec{i}},\underline{\boldvec{y}})$ and obtain $\neg C(i,d)$ or $\neg C(i',d')\vee \neg C(j',e')$.
For each pair $(i,d)$, we introduce a new variable $x_{i,d}$ to represent the value of $C(i,d)$.
We then define $I_n$ to be the set $\{(g,(x_{i,d},x_{j,e}))\mid i,j\in[n],(i,j)\in E_G, d=e\}$, where $g(x,y) = OR(\bar{x},\bar{y})$ for two variables $x$ and $y$ and $\bar{x}$ (resp., $\bar{y}$) denotes the negation of $x$ (resp., $y$).
It then follows that $\Phi$ is satisfiable iff $I_n$ is in $BCSP_2(\Gamma_{OR})$. This implies that $\mathrm{2COLOR}$ is $\dl$-m-reducible to $\mathrm{BCSP}_2(\Gamma_{OR})$.
\end{yproof}

%%%%%%
%%%%%%
\section{Optimization Variant of SNL}

We turn our attention to another variant of SNL in the field of optimization problems.
In general, an \emph{optimization problem} has the form $(op,I,S,cost)$, where $op\in\{\max,\min\}$, $I$ is a set of instances, $S$ is a set of (feasible) solutions, an $cost:I\times S\to\nat$ is a (partial) cost function. As a concrete example of optimization problems,
$\mathrm{MAX}\mbox{-}\mathrm{2SAT}$ asks to find the truth assignment of a given 2CNF Boolean formula that maximizes the number of satisfying clauses of the formula.
Despite a wide-range of studies on $\np$ optimization problems, there have been a few  works on logarithmic-space optimization \cite{Tan07,Yam13,Yam13b}.

%%%%
\subsection{Maximal SNL (or MAXSNL)}\label{sec:maximal}

Papadimitriou and Yannakakis \cite{PY91} were the first to study the computational complexity of an optimization version of $\snp$, called $\mathrm{MAXSNP}$. In a similar vein, we intend to study an optimization version of $\snl$ for   promoting the better understandings of $\snl$.
Along this line of studies, we further explore $\snl$ optimization problems and log-space approximation schemes based on $\snl$.

\begin{definition}\label{def-MAXSNL}
We define $\mathrm{MAXSNL}$ to be composed of all maximization problems that satisfy the following condition: there exist $\snl$ sentences $\Phi$ of the form given in Definition \ref{SNL-sentence} with relational and domain structures $\SSS$ and $\DD$ for $\Phi$ for which
each maximization problem asks to find a solution $\boldvec{P} = (P_1,P_2,\ldots,P_l)$ satisfying $\bigwedge_{i=1}^{l} Func(P_i)$ that maximizes the value
$\Pi(\boldvec{\underline{P}}) =
|\{(\boldvec{\underline{i}},\boldvec{\underline{y}})\mid \bigwedge_{j=1}^{t}\psi_j(\boldvec{\underline{P}}, \boldvec{\underline{i}},\boldvec{\underline{y}}, \boldvec{\underline{S}}, \boldvec{\underline{c}}) \}|$ of the objective function,
where $\psi_j$, $\boldvec{P}, \boldvec{i}, \boldvec{y},\boldvec{S}, \boldvec{c}$
are sequences of variables, predicate and constant symbols appearing in $\Phi$, provided that $\boldvec{i}, \boldvec{y}$ are only first-order variables that appear in $\bigwedge_{j=1}^{t}\psi_j$ and $\boldvec{\underline{i}}, \boldvec{\underline{y}}, \boldvec{\underline{P}}, \boldvec{\underline{S}}, \boldvec{\underline{c}}$ are elements in $\SSS$ and $\DD$ associated respectively with $\boldvec{i}, \boldvec{y}, \boldvec{P}, \boldvec{S}, \boldvec{c}$.
It is important to remember that each $\psi_j$ must satisfy the second-order variable requirements (i)--(ii).
\end{definition}

It is obvious by definition that $\mathrm{MAXSNL}_{\omega} \subseteq \mathrm{MAXSNL} \subseteq \mathrm{MAXSNP}$. Let us see three simple examples of problems in $\mathrm{MAXSNL}$.
The parameterized decision problem $(\mathrm{2SAT},m_{ver})$ was shown  in \cite{Yam17a} to be in  $\para\snl$  by constructing an appropriate $\snl$ sentence for $(\mathrm{2SAT},m_{ver})$, where $m_{ver}(\phi)$ denotes the total number of variables appearing in a 2CNF Boolean formula $\phi$. From this $\snl$ sentence, by carefully eliminating the presence of $m_{ver}$, we can conclude that $\mathrm{MAX}\mbox{-}\mathrm{2SAT}$ belongs to  $\mathrm{MAXSNL}$.
As another concrete  example of problems in $\mathrm{MAXSNL}$, we consider $\mathrm{MAX}\mbox{-}\mathrm{CUT}$ whose goal is to find a set $S$ of vertices of a given undirected graph $G=(V,E)$ for which the number of edges crossing between $S$ and  $V-S$ is maximized.

\begin{yexample}\label{MAX-CUT-in-SNL}
{\rm
As an instance of $\mathrm{MAX}\mbox{-}\mathrm{CUT}$, we take an arbitrary undirected graph $G=(V_G,E_G)$ with $V_G=[n]$ and $E_G\subseteq [n]\times[n]$ for a number $n\in\nat^{+}$.
We introduce a predicate symbol $E$ for which $E(i,j)$ means that $(i,j)$ is an edge in $E_G$.
Given a second-order functional variable $P$, $P(i,1)$ (resp, $P(i,0)$) indicates that vertex $i$ belongs to a solution set $S$ (resp., $V-S$).
Consider the following $\snl$ sentence indicating the existence of such a set $S$: $\Phi\equiv \exists^fP\forall i\forall j [\phi_1(P,i)\wedge \phi_1(P,j) \wedge (E(i,j)\to \phi_2(P,E,i,j))]$, where
$\phi_1\equiv  (P(i,0)\vee P(i,1))\wedge \neg(P(i,0)\wedge P(i,1))$ and $\phi_2\equiv  (P(i,1)\wedge P(j,0))\vee (P(i,0)\wedge P(j,1))$. Intuitively, $\phi_1(P,i)$ means that $P$ forms a ``function'' on input $i$ and $\phi_2(P,E,i,j)$ means that either one of $i$ and $j$ (or both) belongs to $S$, provided that $(i,j)\in E_G$.

Since $Func(P)\equiv \forall i[\phi_1(P,i)]$, it then follows that $\Phi$ is true iff $\exists^fP\forall i\forall j[Func(P)\wedge \phi_2(P,E,i,j)]$. Thus, $\Phi$ is also an $\snl_{\omega}$ sentence.
The objective function $\Pi(P)$ for $\mathrm{MAX}\mbox{-}\mathrm{CUT}$  is expressed as $|\{(i,j)\mid \phi_1(P,i)\wedge \phi_1(P,j)\wedge E(i,j)\wedge \phi_2(P,E,i,j)\}|$.
We thus conclude that $\mathrm{MAX}\mbox{-}\mathrm{CUT}$ belongs to $\mathrm{MAXSNL}$.
}
\end{yexample}

We discuss another simple example, called $\mathrm{MAX}\mbox{-}\mathrm{UK}$, which is a maximization version of $\mathrm{UK}$ (discussed in Example \ref{example-UK}) of the following specific form: one asks  to find a subset $S\subseteq[n]$ that maximizes the value $\sum_{i\in S}a_i$, not exceeding  the upper bound $b$, for any given series $(1^b,1^{a_1},1^{a_2},\ldots,1^{a_n})$ of unary strings with $b,a_1,a_2,\ldots,a_n\in\nat^{+}$.

\begin{yexample}\label{example:MAXUK}
{\rm
Now, we assert that $\mathrm{MAX}\mbox{-}\mathrm{UK}$ belongs to  $\mathrm{MAXSNL}$. To verify this assertion, we recall the notation from Example \ref{example-UK}. We then introduce a second-order variable $P$ so that $P(k,w)$ means the equality $w = \sum_{i\in S\cap [k]}a_i$ for a  certain fixed solution $S$ ($\subseteq [n]$), provided that each instance $x$ of $\mathrm{MAX\mbox{-}UK}$ has the form $(1^b,1^{a_1},1^{a_2},\ldots,1^{a_n})$.
For convenience, we set $\sum_{i\in S\cap[k]}a_i$ to be $0$ whenever $k=0$.
Following Example \ref{example-UK}, we define two formulas:
$\phi_1\equiv 0\leq i<n\wedge 0<j\leq t \wedge 0\leq s+t\leq b \wedge P(0,0)\wedge P(i,s)\wedge P(i+1,s+t)$ and
$\phi_2\equiv t=0\vee I(i+1,t)$, where the new supplemental variable $j$ is meant to ``count'' the number of  elements $(i,0,s,t),(i,1,s,t),\ldots, (i,t-1,s,t)$ whose variables $i,s,t$ satisfy the above formulas $\phi_1$ and $\phi_2$.
The objective function $\Pi(\underline{P})$ for a solution $\underline{P}$ of a $\mathrm{MAX}\mbox{-}\mathrm{UK}$ instance $x$ is then set to be $|\{ (i,j,s,t)\mid \phi_1\wedge \phi_2 \}|$,
Notice that the objective function $\Pi(\underline{P})$ computes $\sum_{i\in S}a_i$ if $\underline{P}(k,w)$ indicates $w=\sum_{j\in S\cap[k]}a_j$. It follows that $\sum_{i\in S}a_i$ is the maximum within $b$ iff $\Pi(\underline{P})$ is the maximum. This places $\mathrm{MAX}\mbox{-}\mathrm{UK}$ to $\mathrm{MAXSNL}$.
}
\end{yexample}

As another formulation of $\mathrm{NAX\mbox{-}UK}$, we make $P(i,b)$ indicate that we choose the $i$th element if $b=1$ and do not choose any element if $b=0$. Let $\phi\equiv 0\leq i\leq n \wedge 0<j\leq z\wedge [(P(i,1)\wedge I(i,z))\vee (P(i,0)\wedge z=0)]$. Nevertheless, the formulation given in Example \ref{example:MAXUK} will be used in Section \ref{sec:subclass-APXL}. 
For any two optimization problems in $\mathrm{MAXSNP}$, a special reduction, called (polynomial-time) linear reduction, was introduced in \cite{PY91}.
Concerning log-space computing, we instead use the notion of logarithmic-space AP-reducibility \cite{Yam13b}.
Given two optimization (i.e., either maximization or minimization) problems $P_1$ and $P_2$, we say that $P_1$ is \emph{logarithmic-space (or log-space) AP-reducible to} $P_2$ if there are two constants $c_1,c_2>0$ and two  functions $f$ and $g$ in $\fl$ such that,
for any value $r\in\rational^{>1}$,
(i) for any instance $x$ of $P_1$, $f(x,r)$ is an instance of $P_2$ and
(ii) for any solution $s$ to the instance $f(x,r)$ of $P_2$, $g(x,s,r)$ is  a solution of the instance $x$ of $P_1$ satisfying $err(x,g(x,s,r)) \leq c_2 \cdot err(f(x,r),s)$. Here, $err(u,z)$ denotes the value  $\max\{\frac{cost(opt(u))}{cost(z)},\frac{cost(z)}{cost(opt(u))}\}-1$ for strings $u$ and $z$, assuming that these denominators are not zero, where $opt(u)$ means an optimal solution to instance $u$ and $cost(z)$ means the value (or cost) of string $z$.
To distinguish it from (standard) $\dl$-m-reductions, we use the special notation of $\leq^{\dl}_{AP}$ to mean these log-space AP-reductions.

\begin{lemma}\label{reduction-composite}
Let $r$ denote any nondecreasing function from $\nat$ to $\nat$. Given two optimization problems $\Xi_1$ and $\Xi_2$, if $\Xi_1\leq^{\dl}_{AP}\Xi_2$ and $\Xi_2$ is log-space approximable within ratio $r(n)$, then $\Xi_1$ is also log-space approximable within ratio $O(r(n^t+t))$ for a certain fixed constant $t\geq 1$.
\end{lemma}

\begin{yproof}
Let us take two constants $c_1,c_2>0$ and two functions $f,g\in\fl$ that make $\Xi_1$ log-space AP-reducible to $\Xi_2$.
Given a nondecreasing function $r$, we take another function $h$ in $\fl$ such that, for any $x$, $h(x)$ is an approximate solution to the instance $x$ of $\Xi_2$ within approximation ratio $r(|x|)$.
In what follows, we intend to construct a function $k$ that produces an approximate solution to each instance of $\Xi_1$.

Consider the composite function $k=g\circ h\circ f$. For any instance $x$, since $f(x)$ is an instance of $\Xi_2$, $h\circ f(x)$ is a solution to the instance $f(x)$. Thus, $k(x)$ is an approximate solution to the instance $x$ of $\Xi_1$. It then follows that $cost(opt(x))\leq c_1\cdot cost(opt(f(x)))$ and $err(k(x),opt(x))\leq c_2 \cdot err(h\circ f(x),opt(f(x)))$.
Since the approximation ratio $r(|z|)$ for $\Xi_2$ equals $err(h(z),opt(z))+1$ for any instance $z$ to $\Xi_2$,
it follows that $err(h\circ f(x),opt(f(x)))$ equals $r(|f(x)|)-1$. Thus, the value $err(k(x),opt(x))+1$ is upper-bounded by $c_2(r(|f(x)|)-1)+1$. Since $|f(x)|\leq |x|^t+t$ for a certain constant $t\in\nat^{+}$, $err(k(x),opt(x))+1$ is at most $c_2\cdot r(|x|^t+t)+1$. This implies that $\Xi_1$ is approximable within ratio $c_2r(|x|^t+t)+1$.
\end{yproof}

It is important to note that every minimization problem can be log-space AP-reducible to its associated maximization problem \cite{Yam13,Yam13b}.
See also \cite{PY91} for a similar result in the polynomial-time setting.  This fact helps us focus only on maximization problems in the following statement.

\begin{proposition}\label{complexity-equality}
$\leq^{\dl}_{AP}\!(\mathrm{MAXSNP}) = \;
\leq^{\dl}_{AP}\!(\mathrm{MAXSNL}) = \; \leq^{\dl}_{AP}\!(\mathrm{MAXSNL}_{\omega})$.
\end{proposition}

Papadimitriou and Yannakakis \cite{PY91} demonstrated that every maximization problem in $\mathrm{MAXSNP}$ can be polynomial-time linear reducible to $\mathrm{MAX}\mbox{-}\mathrm{3SAT}$. Their reduction is in fact carried out using only log space. This fact immediately implies that every maximization problem in $\mathrm{MAXSNP}$ is log-space AP-reducible to $\mathrm{MAX}\mbox{-}\mathrm{3SAT}$. In other words, $\mathrm{MAX}\mbox{-}\mathrm{3SAT}$ is complete for $\mathrm{MAXSNP}$ under log-space AP-reductions. For MAXSNL, we can show the following completeness claim for $\mathrm{MAX\mbox{-}CUT}$.

\begin{lemma}\label{MAX-CUT-complete}
$\mathrm{MAX}\mbox{-}\mathrm{CUT}$ is complete for $\mathrm{MAXSNL}$ under log-space AP-reductions. It is also possible to replace $\mathrm{MAXSNL}$ by $\mathrm{MAXSNL}_{\omega}$.
\end{lemma}

\begin{yproof}
Let us recall from Example \ref{MAX-CUT-in-SNL} that $\mathrm{MAX}\mbox{-}\mathrm{CUT}$ falls in $\mathrm{MAXSNL}$. Consider any maximization problem $\Xi$ in $\mathrm{MAXSNL}$ with an associated $\snl$  sentence $\Phi$ of the form $\exists^fP_1\cdots \exists^fP_l\forall \boldvec{i} \forall \boldvec{y}[\bigwedge_{j=1}^{t}\phi_j(P_1,\ldots,P_l,\boldvec{i},\boldvec{y})]$.  We take three steps to construct a log-space AP-reduction from $\Xi$ to $\mathrm{MAX}\mbox{-}\mathrm{CUT}$.
Firstly, we reduce $\Xi$ to $\mathrm{MAX}\mbox{-}\mathrm{2SAT}$ and then  reduce $\mathrm{MAX}\mbox{-}\mathrm{2SAT}$ to $\mathrm{MAX}\mbox{-}\mathrm{WTDCUT}$. Here, $\mathrm{MAX}\mbox{-}\mathrm{WTDCUT}$ is a ``weighted'' version of $\mathrm{MAX}\mbox{-}\mathrm{CUT}$, which is obtained by allowing each edge to hold a (positive integer) weight and maximizing the total weight of edges whose endpoints are assigned to two different sets $S$ and $V-S$.
Finally, we reduce $\mathrm{MAX}\mbox{-}\mathrm{WTDCUT}$ to $\mathrm{MAX}\mbox{-}\mathrm{CUT}$ to complete the proof.

(1) When a domain structure $\DD_x$ for $\Phi$ is given for an instance $x$, the variable tuple $(\boldvec{i},\boldvec{y})$ takes only polynomially many different values $(\boldvec{\underline{i}},\boldvec{\underline{y}})$.
We assign those values $(\boldvec{\underline{i}},\boldvec{\underline{y}})$ to $(\boldvec{i},\boldvec{y})$ one by one to generate polynomially many ``formulas'' $\phi_j(P_1,\ldots,P_l,\boldvec{\underline{i}},\boldvec{\underline{y}})$. Notice that each $\phi_j$ can be rewritten as a formula made up from  variables of the form $P_k(\underline{i},\underline{y})$ or $\neg P_k(\underline{i},\underline{y})$ as well as the constants $T$(true)  and $F$ (false) because $S_r(\boldvec{\underline{i}},\boldvec{\underline{y}})$'s can be evaluated to be either $T$ or $F$. The second-order variable requirement of $\Phi$ forces this formula $\phi_j$ to be expressed as a 2CNF formula. Thus, $\Xi$ is reduced to $\mathrm{MAX}\mbox{-}\mathrm{2SAT}$.

(2) The reduction $\mathrm{MAX}\mbox{-}\mathrm{2SAT} \leq^{\dl}_{AP}\mathrm{MAX}\mbox{-}\mathrm{WTDCUT}$ is constructed as follows. We loosely follow an argument of \cite[Theorem 2]{PY91}. Let $\phi$ be any 2CNF Boolean formula of the form $\bigwedge_{i=1}^{t}\phi_i$ with $\phi_i\equiv z_{i,1}\vee z_{i,2}$ for certain literals $z_{i,1}$ and $z_{i,2}$. We then  construct a weighted undirected graph $G=(V_G,E_G)$. The vertices are labeled with variables as well as their negations except for a special vertex $w$. Sequentially, we choose a clause $\phi_i$ and then add three edges to form a triangle among three vertices $z_{i,1},z_{i,2},w$. Moreover, we add an edge between every variable $x$ and its negation $\bar{x}$ with weight of $2k$, where $k$ is the number of times that either $x$ or $\bar{x}$ appears in $\phi$. The weight of any edge in   each triangle $(z_{i,1},z_{i,2},w)$ is $2k'$, where $k'$ is the number of clauses in which the pair $z_{i,1},z_{i,2}$ appears simultaneously (ignoring their appearance order).
Notice that any edge weight is always even.
This modification can be done in log space. As argued in \cite{PY91}, the objective value is twice as large as the sum of the number of literal occurrences and the number of satisfying clauses.

(3) The third reduction $\mathrm{MAX}\mbox{-}\mathrm{WTDCUT} \leq^{\dl}_{AP}\mathrm{MAX}\mbox{-}\mathrm{CUT}$ is shown as follows.
From the argument of (2), it suffices to consider the case where the weight of each edge of an undirected graph is an even number. For each edge $(v_1,v_2)$ with weight $2k$, we prepare $k$ new vertices, say, $u_1,u_2,\ldots, u_k$ and add two edges $(v_1,u_i)$ and $(u_i,v_2)$ for each index $i\in[k]$.

The second part of the lemma follows immediately by analyzing (1)--(3).
\end{yproof}

Proposition \ref{complexity-equality} follows from Lemma \ref{MAX-CUT-complete} since $\mathrm{MAX}\mbox{-}\mathrm{CUT}$ is in $\mathrm{MAXSNL}$ by Example \ref{MAX-CUT-in-SNL}.

%%%
%%%

\vs{2}
\begin{proofof}{Proposition \ref{complexity-equality}}
Since  $\mathrm{MAXSNL}_{\omega}\subseteq \mathrm{MAXSNL}\subseteq \mathrm{MAXSNP}$, this fact instantly implies that $\leq^{\dl}_{AP}\!(\mathrm{MAXSNL}_{\omega})\subseteq \; \leq^{\dl}_{AP}\!(\mathrm{MAXSNL}) \subseteq \; \leq^{\dl}_{AP}\!(\mathrm{MAXSNP})$.
To see another inclusion, as noted earlier, every maximization problem in $\mathrm{MAXSNP}$ is  log-space AP-reducible to $\mathrm{MAX}\mbox{-}\mathrm{3SAT}$. This fact implies that $\mathrm{MAXSNP} \subseteq \; \leq^{\dl}_{AP} \!(\mathrm{MAX}\mbox{-}\mathrm{3SAT})$.

Next, we consider the maximization problem $\mathrm{MAX}\mbox{-}\mathrm{CUT}$, which belongs to $\mathrm{MAXSNL}_{\omega}$ by Example \ref{MAX-CUT-in-SNL}.
It then suffices to show that (*) $\mathrm{MAX}\mbox{-}\mathrm{3SAT}\leq^{\dl}_{AP} \mathrm{MAX}\mbox{-}\mathrm{CUT}$
because  we obtain $\leq^{\dl}_{AP}\!(\mathrm{MAX}\mbox{-}\mathrm{3SAT}) \subseteq \; \leq^{\dl}_{AP}\!(\mathrm{MAX}\mbox{-}\mathrm{CUT}) \subseteq \; \leq^{\dl}_{AP}\!(\mathrm{MAXSNL}_{\omega})$ from Lemma \ref{MAX-CUT-complete}.
The desired reduction (*) will be achieved by proving (**) $\mathrm{MAX}\mbox{-}\mathrm{3SAT}\leq^{\dl}_{AP} \mathrm{MAX}\mbox{-}\mathrm{2SAT}$ because $\mathrm{MAX}\mbox{-}\mathrm{2SAT}\leq^{\dl}_{AP} \mathrm{MAX}\mbox{-}\mathrm{CUT}$ has already been shown in the proof of Lemma \ref{MAX-CUT-complete}.
Our target reduction (**) comes from a note of Williams \cite{Wil04} on
a transformation of a clause $\phi$ of the form $z_1\vee z_2\vee z_3$, where $z_1,z_2,z_3$ are literals, to the following ten clauses: $z_1\vee z_1$, $z_2\vee z_2$, $z_3\vee z_3$, $w\vee w$, $\overline{z_1}\vee \overline{z_2}$, $\overline{z_2}\vee \overline{z_3}$, $\overline{z_1}\vee \overline{z_3}$, $z_1\vee \overline{w}$, $z_2\vee\overline{w}$, $z_3\vee \overline{w}$, where $w$ is a new variable associated with $\phi$ and $\overline{z_i}$ ($i\in\{1,2,3\}$) denotes $\overline{x}$ (resp., $x$) if $z_i$ is a variable $x$ (resp., if $z_i$ is the negation of a variable $x$). As noted in  \cite{Wil04}, for any given assignment, (i) if it satisfies $\phi$, then exactly $7$ out of ten classes are satisfied and (ii) if it does not satisfy $\phi$, then exactly $6$ out of ten classes are satisfied.
\end{proofof}

%%%

In the end of this subsection, we briefly argue a relationship to the work of Bringman et al. \cite{BCFK21}, who discussed a subclass of $\mathrm{MAXSNP}$, called $\mathrm{MAXSP}$.
As a concrete example, let us consider a typical  maximization problem in $\mathrm{MAXSP}$, known as $\mathrm{MAX}\mbox{-}\mathrm{IP}$ \cite{BCFK21}, in which one asks to find a pair $(x_1,x_2)$ in $X_1\times X_2$ that maximizes the value $|\{j\in[d(n)] \mid x_1[j]\cdot x_2[j]=1\}|$ for two given sets $X_1,X_2\subseteq \{0,1\}^{d(n)}$ with $|X_1|=|X_2|=n$, where $d(n)=n^{\gamma}$ for a small constant $\gamma>0$ and $x[j]$ denotes the $j$th bit of $x$.
We show that this maximization problem is indeed in $\mathrm{MAXSNL}$.

\begin{lemma}
$\mathrm{MAX}\mbox{-}\mathrm{IP}$ belongs to $\mathrm{MAXSNL}$.
\end{lemma}

\begin{yproof}
To see this, let us consider the sentence $(\exists^fP)(\forall i) [ \Phi_1(P,\tilde{X}_1,\tilde{X}_2)\wedge \Phi_2(P,Bit,i)]$, where $\tilde{X}_1$, $\tilde{X}_2$, and $Bit$  are predicate symbols for which each $\tilde{X}_k(z)$ ($k\in\{1,2\}$) means that $z$ is in set $X_k$ and $Bit(z,i)$ means that the $i$th bit of $z$ is $1$.
The two formulas $\Phi_1$ and $\Phi_2$ are defined as follows: $\Phi_1\equiv (\forall z_1,z_2)[P(1,z_1)\wedge P(2,z_2)\to \tilde{X}_1(z_1)\wedge \tilde{X}_2(z_2)]$ and $\Phi_2\equiv (\forall z_1,z_2)[ P(1,z_1)\wedge P(2,z_2) \to Bit(z_1,i) \wedge Bit(z_2,i)]$.
It is clear that the objective function $\Pi(P)$ is expressed as $\Pi(P)= |\{i\mid \Phi_1\wedge \Phi_2\}|$.
\end{yproof}

%%%%
\subsection{Subclass of MAXSNL in APXL}\label{sec:subclass-APXL}

We discuss the approximability of optimization problems in $\mathrm{MAXSNL}$.  
In the polynomial-time setting, it is known from \cite[Theorem 1]{PY91} that every optimization problem in $\mathrm{MAXSNP}$ falls in APX; namely, it is approximated in polynomial time within a certain fixed approximation ratio.
The log-space approximability of optimization problems was studied earlier in \cite{Tan07,Yam13,Yam13b}.
Following \cite{Yam13,Yam13b}, we use the notation $\mathrm{APXL}$ (also denoted $\mathrm{APXL}_{\mathrm{NLO}}$ in \cite{Yam13b}) for the collection of $\nl$ optimization problems that can be approximated in polynomial time using log space with fixed constant approximation ratios.
Several NL optimization problems are known to be ``complete'' for APXL under various reductions. Those complete problems include the maximization binary 2-bounded close-to-unary knapsack problem and the maximum fixed-length $\lambda$-nondeterministic finite automata problem \cite{Yam13b}.
Since $\mathrm{MAXSNL}$ contains optimization problems of extremely high complexity by Proposition \ref{complexity-equality}, it seems difficult to show that every problem in $\mathrm{MAXSNL}$ is in $\mathrm{APXL}$.
Hence, it is natural to ask what subclass of $\mathrm{MAXSNL}$ makes its problems fall in $\mathrm{APXL}$.

As an example of such problem, we first recall $\mathrm{MAX\mbox{-}UK}$ from Example \ref{example:MAXUK}.

\begin{proposition}\label{MAXUK-approximation}
$\mathrm{MAX}\mbox{-}\mathrm{UK}$ is in $\mathrm{APXL}$.
\end{proposition}

\begin{yproof}
Let $x$ denote  any instance of the form $(1^b,1^{a_1},1^{a_2},\ldots,1^{a_n})$ given to $\mathrm{MAX}\mbox{-}\mathrm{UK}$. For simplicity, we assume that $0<a_i\leq b$ for all $i\in[n]$.
Consider the following simple, greedy algorithm. It is possible to enumerate the integers $a_1,a_2,\ldots,a_n$ given in the unary representation in the descending order using log space by scanning each input symbol of $x$ back and forth (see, e.g., \cite{Yam23}). Thus, we can assume without  loss of generality that $a_1\geq a_2\geq \cdots \geq a_n$ in the rest of our argument. We wish to inductively construct a subset $S$ of $[n]$,  starting with $S=\setempty$, by executing the following scheme.
By incrementing $i$ by one from $i=0$, we pick $a_i$ and check whether  $\sum_{j\in S}a_j+a_i\leq b$. If so, then we expand $S$ by adding $i$ to $S$; otherwise, we do nothing. After reading all $a_i$'s, we output the value $\sum_{j\in S}a_i$.

Next, we intend to verify that the above scheme is indeed an approximation scheme with an approximation ratio of at most $2$. Let $S=\{i_1,i_2,\ldots,i_k\}$ be the set constructed by this approximation scheme with $1\leq i_1< i_2 < \cdots < i_k\leq n$.
In the case of $\sum_{i=1}^{n}a_i\leq b$, $S$ must be $[n]$, and thus the scheme correctly solves the given problem.
In the other case of $\sum_{i=1}^{n}a_i>b$, we wish to prove that $\frac{b}{2} \leq \sum_{i\in S}a_i \leq b$. To prove this, on the contrary, we first assume that $\sum_{i\in S}a_i<\frac{b}{2}$.
If $i_k<n$, then we obtain $a_{i_k+1}>\frac{b}{2}$ because, otherwise, we obtain $\sum_{i\in S}a_i+a_{i_k+1}\leq b$ and thus the algorithm should  place $i_{k}+1$ in $S$, a contradiction to the definition of $S$.
Since $a_{i_1}\geq a_{i_2}\geq \cdots \geq a_{i_k}\geq a_{i_k+1}$, it follows that $a_{i_j}>\frac{b}{2}$ for all indices $j\in[k]$. This is in contradiction with $\sum_{i\in S}a_i<\frac{b}{2}$.
By contrast, let us consider the case of $i_k=n$. Since $i_k=n$ and $\sum_{i\in S}a_i<\frac{b}{2}$, we conclude that $a_i<\frac{b}{2}$ for all $i\in[n]$. Since $\sum_{i\in [n]}a_i>b$, $S\neq [n]$ follows. We then take the smallest number $i'$ in $[n]-S$. This means that the algorithm does not choose $a_{i'}$. Hence, it follows that $\sum_{i\in S\cap[i']}a_i + a_{i'}>b$. From this, we obtain $a_{i'}>b-\sum_{i\in S\cap[i']}a_i \geq b-\frac{b}{2} = \frac{b}{2}$. This contradicts the upper bound $a_i<\frac{b}{2}$ for all $i\in[n]$.
 Therefore, $\sum_{i\in S}a_i\geq \frac{b}{2}$ must hold. Since the optimal solution has a value of at most $b$, the approximation ratio cannot exceed $2$.
\end{yproof}

%%%

If we denote the ``optimal'' choice of $\boldvec{P}$ (i.e., an \emph{optimal solution}) by $\boldvec{P}_{opt}$, then $\Pi(\boldvec{\underline{P}})\leq \Pi(\boldvec{\underline{P_{\mathit{opt}}}})$ holds for any $\boldvec{\underline{P}}$.

%%%%

Now, let us introduce a new subclass of $\mathrm{MAXSNL}$, called  $\mathrm{MAX}\tau\mathrm{SNL}$, whose elements all fall in $\mathrm{APXL}$.
We first recall an SNL sentence $\Phi$ given in Definition \ref{def-MAXSNL} and, for a variable sequence $\boldvec{P}=(P_1,P_2,\ldots,P_l)$,
its associated objective function $\Pi(\boldvec{\underline{P}})$ that computes the value
$|\{ (\boldvec{\underline{i}},\boldvec{\underline{y}}) \mid \bigwedge_{j=1}^{t} \psi_j (\boldvec{\underline{P}},\boldvec{\underline{i}},\boldvec{\underline{y}}, \boldvec{\underline{S}},\boldvec{\underline{c}})\}|$, where   $\boldvec{{i}}=({i_1},\ldots,{i_{r}})$, $\boldvec{{y}}=({y_1},\ldots,{y_{s}})$,  $\boldvec{{S}}=({S_1},\ldots,{S_d})$, and $\boldvec{{c}}=({c_1},\ldots,{c_{d'}})$ are sequences of variables, predicate and constant symbols.
As noted in Definition \ref{def-MAXSNL}, only first-order variables in $\boldvec{i}$ and $\boldvec{y}$ appear in $\bigwedge_{j=1}^{t}\psi_j$.

By the second-order variable requirement (ii), each $\psi_j$ can be rewritten as finite disjunctions where at most two such disjuncts have the form $(\bigwedge_{k,i,\boldvec{v}}P_k(i,\boldvec{v}))\wedge (\bigwedge_{k',i',\boldvec{v'}}\neg P_{k'}(i',\boldvec{v'})) \wedge R$ for a first-order quantifier-free subformula $R$.
Toward the introduction of $\mathrm{MAX}\tau\mathrm{SNL}$, we further place a restriction that the formula $\bigwedge_{j=1}^{t}\psi_j$ has the following special form:
\[
(*) \hs{5} (\bigwedge_{i=1}^{l} P_k(e,\boldvec{u}_i))\wedge (\bigwedge_{i=1}^{l} P_{k}(e+1,\boldvec{v}_i)) \wedge R^{(-)}(e+1, \boldvec{{i}},\boldvec{{y}}, \boldvec{{u_1}},\ldots,\boldvec{{u_l}}, \boldvec{{v_1}},\ldots,\boldvec{{v_l}},  \boldvec{{S}},\boldvec{{c}})
\]
for a first-order quantifier-free SNL formula $R^{(-)}$ having no second-order variables.

Using $R^{(-)}$, we then define $\Pi(\boldvec{P})\lceil_a$ to be the value $|\{ (\underline{e},\boldvec{\underline{i}},\boldvec{\underline{y}})\mid 0<e\leq a\wedge (\bigwedge_{j=1}^{l} P_j(e,\boldvec{u_j}))\wedge (\bigwedge_{j=1}^{l} P_j(e+1,\boldvec{v_j})) \wedge  R^{(-)}(\underline{e}+1, \boldvec{\underline{i}},\boldvec{\underline{y}}, \boldvec{\underline{u_1}},\ldots,\boldvec{\underline{u_l}}, \boldvec{\underline{v_1}},\ldots,\boldvec{\underline{v_l}},  \boldvec{\underline{S}},\boldvec{\underline{c}})\}|$.
We remark that, for any $\boldvec{\underline{P}}$, $\Pi(\boldvec{\underline{P}})\lceil_0=0$, $\Pi(\boldvec{\underline{P}})\lceil_a\leq \Pi(\boldvec{\underline{P}})\lceil_{a+1}$, $\Pi(\boldvec{\underline{P}})\lceil_{n} =\Pi(\boldvec{\underline{P}})$ for all $a\in[0,n-1]_{\integer}$.
For convenience, we also define  $g_a(\boldvec{\underline{u_1}},\ldots,\boldvec{\underline{u_l}}, \boldvec{\underline{v_1}}, \ldots,\boldvec{\underline{v_l}})$
to be the value $|\{(\boldvec{\underline{i}},\boldvec{\underline{y}})\mid R^{(-)}(a,  \boldvec{\underline{i}},\boldvec{\underline{y}}, \boldvec{\underline{u_1}},\ldots,\boldvec{\underline{u_l}}, \boldvec{\underline{v_1}},\ldots,\boldvec{\underline{v_l}},  \boldvec{\underline{S}},\boldvec{\underline{c}})\}|$.
It then follows that, if $\bigwedge_{j=1}^{l}P_j(e,\boldvec{u}_j)$ and $\bigwedge_{j=1}^{l}P_j(e+1,\boldvec{v}_j)$ are true, then $\Pi(\boldvec{\underline{P}})\lceil_{a+1} = \Pi(\boldvec{\underline{P}})\lceil_{a}\: + \: g_{a+1}(\boldvec{\underline{u_1}},\ldots, \boldvec{\underline{u_l}},\boldvec{\underline{v_1}},
\ldots,\boldvec{\underline{v_l}})$.
This seems to provide a simple, easy way to compute $\Pi(\boldvec{\underline{P}})\lceil_{a+1}$ from $\Pi(\boldvec{\underline{P}})\lceil_{a}$.

Let us see a concrete example of $\Pi(\boldvec{P})\lceil_a$ in Example \ref{example-Rminus}.

\begin{example}\label{example-Rminus}
Consider the maximization problem $\mathrm{MAX\mbox{-}UK}$.
Let us recall its formulation given in Example \ref{example:MAXUK}.
We further define $R^{(-)}$ as follows: $R^{(-)}(i,j,z,u,v) \equiv 0< i<n\wedge 0<j\leq z \wedge u+z=v \wedge 0\leq v \leq b \wedge (z=0 \vee I(i,z))$.
For each value $a\in[n]$, the restriction $\Pi(\underline{P})\lceil_a$ thus has the form $\Pi(\underline{P})\lceil_a = |\{ (i,j,z)\mid 0<i\leq a\wedge P(i,u)\wedge P(i+1,v)\wedge R^{(-)}(i+1,j,z,u,v)\}|$.
Moreover, $g_a(u,v)$ takes the value $|\{(j,z)\mid R^{(-)}(a,j,z,u,v)\}|$. Clearly, when $P(a,u)$ and $P(a+1,v)$ are true, we obtain $\Pi(\underline{P})\lceil_{a+1} = \Pi(\underline{P})\lceil_{a} \:+\: g_{a+1}(u,v)$.
\end{example}

As noted, the use of $g_a$ provides a simple, easy way to compute $\Pi(\boldvec{\underline{P}})\lceil_{a+1}$ from $\Pi(\boldvec{\underline{P}})\lceil_a$. However, this does not seem to immediately guarantees a ``log-space'' procedure of computing  $\Pi(\underline{P})$.
To overcome this difficulty, we wish to expand an underlying vocabulary by appending a new function symbol ``$h$''.
We write ``$h(i)$'' with a variable $i$ to mean the outcome of this function on input $i$. We then syntactically replace each formula ``$P_k(i,\boldvec{u})$'' in $\bigwedge_{j=1}^{t}\psi_j$ by ``$P_k(h(i),\boldvec{u})$''.
Note that, since this is just a syntactical replacement, we keep the same variable sequence $\boldvec{u}$ as ``symbolic'' terms.  %It is important to note that $R^{(-)}$ is not affected by this replacement because $R^{(-)}$ does not contain any $P_k$'s.
For notational convenience, we hereafter write ``$P_k^{(h)}(i,\boldvec{u})$'' instead of ``$P_k(h(i),\boldvec{u})$''. Let $\boldvec{P}^{(h)} =(P^{(h)}_1,P^{(h)}_2,\ldots,P^{(h)}_l)$.
The formula obtained from $\bigwedge_{j=1}^{t}\psi_j$ by this syntactical replacement is succinctly referred to as an \emph{$h$-term expansion of} $\bigwedge_{j=1}^{t}\psi_j$. Naturally, we further obtain $\Pi(\boldvec{\underline{P}}^{(h)})$ and $\Pi(\boldvec{\underline{P}}^{(h)})\lceil_a$ from $\Pi(\boldvec{\underline{P}})$ and $\Pi(\boldvec{\underline{P}})\lceil_a$, respectively.

\begin{example}\label{example-h-term}
We consider an $h$-term expansion of $\Pi(P)\lceil_a$ given in Example \ref{example-Rminus}. Let $h$ denote a newly introduced function symbol, which represents a permutation on $[0,n]_{\integer}$ with $\underline{h}(0)=0$. We introduce $P^{(h)}$ as a new variable associated with $h$ and consider $\Pi(P^{(h)})\lceil_a =|\{ (i,j,z) \mid 0<i\leq a  \wedge P^{(h)}(i,u)\wedge P^{(h)}(i+1,v) \wedge R^{(-)}(h(i+1),j,z,u,v)\}|$ and $g^{(h)}_a(u,v) = |\{(j,z)\mid R^{(-)}(h(a),j,z,u,v)\}|$. It then follows that $\Pi(P^{(h)})\lceil_{a+1} = \Pi(P^{(h)})\lceil_a \:+\: g_{a+1}^{(h)}(u,v)$ if $P^{(h)}(a,u)$ and $P^{(h)}(a+1,v)$ are true.
\end{example}

Now, let us define the complexity class $\mathrm{MAX}\tau\mathrm{SNL}$.

\begin{definition}\label{def-MAXtauSNL}
A maximization problem is in $\mathrm{MAX}\tau\mathrm{SNL}$ if there exists
a quantifier-free SNL formula $R^{(-)}$ having no functional variables, which  naturally induces $\Pi(\boldvec{P})$, $\Pi(\boldvec{P})\lceil_a$, and $g_a$ as shown above, with the following extra three conditions such that the maximization problem asks to find a solution $\boldvec{P}$ that maximizes the value $\Pi(\boldvec{P})$ of the objective function.
Let $a$, $\boldvec{u}_j$, $\boldvec{v}_j$, $\boldvec{P}$, $\boldvec{P'}$, $\boldvec{\hat{P}}$, $\boldvec{S}$, and $\boldvec{c}$ denote sequences of variables, predicates and constant symbols. Let $h$ be any newly introduced function symbol whose interpretation $\underline{h}$ is a permutation on $[0,n]_{\integer}$ with $\underline{h}(0)=0$. Let $A^{(h)}_{a}(\boldvec{\underline{u_1}},\ldots, \boldvec{\underline{u_l}})$ denote the set $\{ g^{(h)}_a(\boldvec{\underline{u_1}},\ldots, \boldvec{\underline{u_l}}, \boldvec{\underline{u'_1}},\ldots, \boldvec{\underline{u'_l}}) \mid \boldvec{\underline{u'_1}},\ldots, \boldvec{\underline{u'_l}}\}$ and     $\boldvec{\underline{u'_1}},\ldots,\boldvec{\underline{u'_l}}$ range over all possible elements.
\begin{enumerate}\vs{-2}
  \setlength{\topsep}{-1mm}%
  \setlength{\itemsep}{1mm}%
  \setlength{\parskip}{0cm}%

\item[(1)] If both $\bigwedge_{j=1}^{l}P_j^{(h)}(a,\boldvec{u_j})$ and $\bigwedge_{j=1}^{l}P_j^{(h)}(a+1,\boldvec{v_j})$ are true, then $\Pi(\boldvec{\underline{P}}^{(h)})\lceil_{a+1} = \Pi(\boldvec{\underline{P}}^{(h)})\lceil_a \: + \: g^{(h)}_{a+1}(\boldvec{\underline{u_1}},\ldots,\boldvec{\underline{u_l}}, \boldvec{\underline{v_1}},\ldots,\boldvec{\underline{v_l}})$.

\item[(2)] If both $\bigwedge_{j=1}^{l} P_j^{(h)}(a,\boldvec{u}_j)$ and $\bigwedge_{j=1}^{l} {\hat{P}_j}^{(h)}(a,\boldvec{v}_j)$ are true, then $\Pi(\boldvec{\underline{P}}^{(h)})\lceil_a \geq \Pi(\boldvec{\underline{\hat{P}}}^{(h)})\lceil_a$ implies $A^{(h)}_{a+1}(\boldvec{\underline{u_1}},\ldots, \boldvec{\underline{u_l}}) \subseteq A^{(h)}_{a+1}(\boldvec{\underline{v_1}},\ldots, \boldvec{\underline{v_l}})$.

\item[(3)]
    (i) For any $\boldvec{P'}$, there exists a $\boldvec{P}^{(h)}$ such that $\Pi(\boldvec{\underline{P'}}) \leq \Pi(\boldvec{\underline{P}}^{(h)})$.
    (ii) For any $\boldvec{P}^{(h)}$, there exists a $\boldvec{P'}$ such that $\Pi(\boldvec{\underline{P}}^{(h)}) \leq \Pi(\boldvec{\underline{P'}})$.
\end{enumerate}
\end{definition}

The symbol ``$\tau$'' in $\mathrm{MAX}\tau\snl$ indicates the ``transitive'' relation of $\Pi(\boldvec{\underline{P}})\lceil_a$ over all values $a$.
The condition (3), in particular, expresses the invariance of the final outcome of $\Pi(\boldvec{\underline{P}}^{(h)})$ over the choice of any permutation $h$.

In what follows, we demonstrate that $\mathrm{MAX}\mbox{-}\mathrm{UK}$ is an example problem of $\mathrm{MAX}\tau\mathrm{SNL}$.
However, it is not clear that $\mathrm{MAX\mbox{-}UK}$ is ``complete'' for $\mathrm{MAX}\tau\mathrm{SNL}$ under naturally chosen reductions.

%%%%

Give a formula $\phi$, we set $[\![\phi]\!]=1$ if $\phi$ is true and $[\![\phi]\!]=0$ if $\phi$ is false.

\begin{lemma}
$\mathrm{MAX\mbox{-}UK}$ is in $\mathrm{MAX}\tau\mathrm{SNL}$.
\end{lemma}

\begin{yproof}
Firstly, let us recall the objective function $\Pi(P)$ defined in Example \ref{example:MAXUK} and $R^{(-)}$, $g_a$, and $\Pi(P)\lceil_a$
defined in Example \ref{example-Rminus} for $\mathrm{MAX}\mbox{-}\mathrm{UK}$.
For the containment of $\mathrm{MAX\mbox{-}UK}$ in $\mathrm{MAX}\tau\mathrm{SNL}$, we argue that this $R^{(-)}$ satisfies the desired conditions (1)--(3) of Definition \ref{def-MAXtauSNL}.
We immediately obtain $\Pi(P)\lceil_{0} = 0$.

Let $x$ denote an instance of the form $(1^b,1^{a_1},1^{a_2},\ldots,1^{a_n})$.
We introduce a new functional variable $P$ indicating that, for each $(e,u_e)$,  $P(e,u_e)$ is true exactly when  $u_e=\sum_{j=1}^{e-1}a_{i_j}\cdot [\![\underline{P}(j,u_j)]\!] + a_{i_e}\cdot z$ for $z\in\{0,1\}$, provided that $P(0,u_0),P(1,u_1),\ldots, P(e-1,u_{e-1})$ are already determined.

Let $h$ denote a newly introduced function symbol, which represents a permutation on $[0,n]_{\integer}$ with $\underline{h}(0)=0$. We introduce $P^{(h)}$ as a new variable associated with $h$ and consider the relevant values  $\Pi(P^{(h)})\lceil_a =|\{ (i,j,z) \mid 0<i\leq a  \wedge P^{(h)}(i,u)\wedge P^{(h)}(i+1,v) \wedge R^{(-)}(h(i+1),j,z,u,v)\}|$ for all $a\in[0,n]_{\integer}$ and $g^{(h)}_a(u,v) = |\{(j,z)\mid R^{(-)}(h(a),j,z,u,v)\}|$.  Note that, if $P^{(h)}(a,u)$ and $P^{(h)}(a+1,u')$ are true, then
$g^{(h)}_{a+1}(u,u')\in \{0,a_{h(a+1)}\}$.

The condition (1) of Definition \ref{def-MAXtauSNL} comes from Example \ref{example-h-term}.
Next, we show the condition (2). Assume that $P^{(h)}(a,u)$, $\hat{P}^{(h)}(a,v)$, $\Pi(P^{(h)})\lceil_a \geq \Pi(\hat{P}^{(h)})\lceil_a$ are true.
Let $A^{(h)}_{a+1}(u) = \{g_{a+1}^{(h)}(u,u')\mid u'\}$ and $A^{(h)}_{a+1}(v)=\{g_{a+1}^{(h)}(v,v')\mid v'\}$, where $u'$ and $v'$ range over all possible values.
It then follows that $A^{(h)}_{a+1}(u),A^{(h)}_{a+1}(v)\subseteq \{0, a_{h(a+1)}\}$. Note that $|A^{(h)}_a(u)|\geq 1$ for all $a$ and $u$.
Since $\Pi(P^{(h)})\lceil_{a}\geq \Pi(\hat{P}^{(h)})\lceil_{a}$, we obtain $A^{(h)}_{a+1}(u)\subseteq A^{(h)}_{a+1}(v)$.

Toward the condition (3), this comes from the fact that the inputs $(1^b,1^{a_1},1^{a_2},\ldots,1^{a_n})$ and its permutated version $(1^b,1^{a_{i_1}},1^{a_{i_2}},\ldots,1^{a_{i_n}})$ with $[n]=\{i_1,i_2,\ldots,i_n\}$ can have the same solutions with the same value of their objective functions.
\end{yproof}

%%%

Finally, we claim that $\mathrm{MAX}\tau\mathrm{SNL}$ is contained in $\mathrm{APXL}$.

\begin{theorem}\label{MAX-tau-SNL}
Every maximization problem in $\mathrm{MAX}\tau\mathrm{SNL}$ belongs to  $\mathrm{APXL}$.
\end{theorem}

\begin{yproof}
Let $D$ denote any maximization problem in $\mathrm{MAX}\tau\mathrm{SNL}$ and consider an appropriately chosen SNL formula $R^{(-)}$ satisfying the aforementioned form (*) with functional variables $\boldvec{P}=(P_1,\ldots,P_l)$ and an objective function $\Pi(\boldvec{P})$ associated with $D$ in Definition \ref{def-MAXtauSNL}.
The restriction $\Pi(\boldvec{P})\lceil_{a}$ is obtained from $\Pi(\boldvec{P})$ by the use of the supplemental function
$g_{a}(\boldvec{u},\boldvec{v}) = |\{ (\boldvec{\underline{i}},\boldvec{\underline{y}}) \mid R^{(-)}(a,\boldvec{\underline{i}}, \boldvec{\underline{y}}, \boldvec{\underline{u}}, \boldvec{\underline{v}}, \boldvec{\underline{S}},\boldvec{\underline{c}}) \}|$ for each $a\in[0,n]_{\integer}$.

To proceed this proof further, we take a new term of the form $h(i)$  and syntactically replace $P(i,\boldvec{u})$ in $\bigwedge_{j=1}^{t}\psi_j$ by $P(h(i),\boldvec{u})$. We introduce a new functional variable $P^{(h)}$ to express $P(h(i),\boldvec{u})$ as $P^{(h)}(i,\boldvec{u})$.

Let $\boldvec{P_{opt}}$ denote an optimal solution of $D$.
Hereafter, we intend to approximate $\boldvec{P_{opt}}$ by defining an appropriate permutation $h$ on $[0,n]_{\integer}$ with $h(0)=0$
and searching for an approximate solution $\boldvec{P}^{(h)}$, which makes $\Pi(\boldvec{P}^{(h)})$ close enough to $\Pi(\boldvec{P_{opt}}^{(h)})$. This is possible because the condition (3) of Definition \ref{def-MAXtauSNL} helps us replace $\Pi(\boldvec{P_{opt}})$ by $\Pi(\boldvec{P_{opt}}^{(h)})$.
To simplify the notation in the subsequent argument, we assume $l=t=1$,  write $P, i,j$ for $\boldvec{P},\boldvec{i},\boldvec{j}$, drop  ``$j$'' from $\psi_j$, and omit $\boldvec{S}$ and $\boldvec{c}$ entirely.
Note that all variables in the tuple  $(i,j,\boldvec{u},\boldvec{v})$ are evaluated as logarithmic-size ``objects''.

Our goal is to determine $P^{(h)}$ that maximizes the value $\Pi(P^{(h)})$ by employing the following greedy approximation algorithm, called $\BB$, which is in essence a generalization of the one given in the proof of Proposition \ref{MAXUK-approximation}.
Initially, we set $s_0=0$, define $h(0)=0$, and take $\boldvec{\underline{v_0}}$ to satisfy $P(h(0),\boldvec{\underline{v_0}})$.
By induction hypothesis, we assume that $s_0,s_1,\ldots,s_{i}, a_0,a_1,\ldots, \boldvec{\underline{v_0}},\ldots, \boldvec{\underline{v_i}}$ are already determined.
Assume also that, for each number $j\in[0,i]_{\integer}$, the value $h(j)$ is already determined. We then define $S_h=\{h(j)\mid j\in[0,i]_{\integer}\}$.
Moreover, we assume that  $P^{(h)}(0,\boldvec{\underline{v_0}}),P^{(h)}(a,\boldvec{\underline{v_1}}), \ldots, P^{(h)}(i,\boldvec{\underline{v_i}})$ are all true.
We then choose $a,\boldvec{\underline{v}}$ with $a\notin S_h$ that maximize the value $g_{a}(\boldvec{\underline{v_i}},\boldvec{\underline{v}})$ over all possible elements for $\boldvec{v}$. Remember that there are only polynomially many possible values assigned to $\boldvec{v}$.
We then define $h(i+1)=a$, include $h(i+1)$ to the set $S_h$, and write  $\boldvec{v_{i+1}}$ for $\boldvec{v}$. We also define $s_{i+1}$ to be $s_{i} \:+\: g_{h(i+1)}(\boldvec{\underline{v_i}},\boldvec{\underline{v_{i+1}}})$ and make $P^{(h)}(i+1,\boldvec{\underline{v_{i+1}}})$ true.
It is not difficult to show that $h$ is indeed a permutation on $[0,n]_{\integer}$ with $h(0)=0$ and that $\Pi(P^{(h)})\lceil_{i+1} = \Pi(P^{(h)})\lceil_{i} \:+\: g^{(h)}_{i+1}(\boldvec{\underline{v_{i}}},\boldvec{\underline{v_{i+1}}})$ holds for all $i\in[0,n-1]_{\integer}$.

Now, we wish to prove by contradiction that $P^{(h)}$ is an approximate solution of $P^{(h)}_{opt}$ with approximation ratio of at most 3; namely, $\Pi(P^{(h)}) \geq \frac{1}{3} \Pi(P^{(h)}_{opt})$.
Toward an intended contradiction, we now assume that $\Pi(P^{(h)}) < \frac{1}{3} \Pi(P^{(h)}_{opt})$.
This immediately yields $\Pi(P^{(h)})\lceil_j< \frac{1}{3} \Pi(P^{(h)}_{opt})$ for all numbers $j\in[0,n]_{\integer}$.
Let us take the smallest number $j_0\in[n]$ such that $\Pi(P^{(h)})\lceil_j < \Pi(P^{(h)}_{opt})\lceil_j$ holds for all $j\geq j_0$. Such a number $j_0$ exists because of $\Pi(P^{(h)})\lceil_{j} < \frac{1}{3} \Pi(P^{(h)}_{opt})$ for all $j\in[0,n]_{\integer}$.
Since the algorithm $\BB$ always chooses elements making the value of $g_a$ the largest, we obtain, in particular, $\Pi(P^{(h)})\lceil_1\geq \Pi(P^{(h)}_{opt})\lceil_{1}$, and thus $j_0>1$ follows.
Hence, we obtain $\Pi(P^{(h)}_{opt})\lceil_{j_0-1} \leq \Pi(P^{(h)})\lceil_{j_0-1}\leq \frac{1}{3} \Pi(P^{(h)}_{opt})$.

If $P^{(h)}_{opt}(j_0-1,\boldvec{\underline{u'}})$ and $P^{(h)}_{opt}(j_0,\boldvec{\underline{v'}})$ are true, then we obtain $\Pi(P^{(h)}_{opt})\lceil_{j_0} = \Pi(P^{(h)}_{opt})\lceil_{j_0-1} \:+\: g^{(h)}_{j_0}(\boldvec{\underline{u'}},\boldvec{\underline{v'}})$  by the condition (1) of Definition \ref{def-MAXtauSNL}.
It then follows that $g^{(h)}_{j_0}(\boldvec{\underline{u'}}, \boldvec{\underline{v'}}) - g^{(h)}_{j_0}(\boldvec{\underline{v_{j_0-1}}}, \boldvec{\underline{v_{j_0}}})
= \Pi(P^{(h)}_{opt})\lceil_{j_0} - \Pi(P^{(h)}_{opt})\lceil_{j_0-1} - ( \Pi(P^{(h)})\lceil_{j_0} - \Pi(P^{(h)})\lceil_{j_0-1})
= (\Pi(P^{(h)}_{opt})\lceil_{j_0} - \Pi(P^{(h)})\lceil_{j_0}) + (\Pi(P^{(h)})\lceil_{j_0-1} - \Pi(P^{(h)}_{opt})\lceil_{j_0-1}) >0$.
This leads to the inequality of  $g^{(h)}_{j_0}(\boldvec{\underline{v_{j_0-1}}},\boldvec{\underline{v_{j_0}}}) < g^{(h)}_{j_0}(\boldvec{\underline{u'}},\boldvec{\underline{v'}})$.

Next, we further claim that   $g^{(h)}_{j_0}(\boldvec{\underline{u'}}, \boldvec{\underline{v'}}) < \frac{1}{3} \Pi(P^{(h)}_{opt})$. Assuming that $g^{(h)}_{j_0}(\boldvec{\underline{u'}}, \boldvec{\underline{v'}})\geq \frac{1}{3} \Pi(P^{(h)}_{opt})$, since the algorithm $\BB$ has chosen $\boldvec{v_2}$ so that  $g_1^{(h)}(\boldvec{v_1},\boldvec{v})$ is the maximum among all possible values $g^{(h)}_a(\boldvec{w},\boldvec{w'})$ for any $(a,\boldvec{w},\boldvec{w'})$, it follows that
$\Pi(P^{(h)})\lceil_1 = \Pi(P^{(h)})\lceil_0 \:+\: g^{(h)}_1(\boldvec{\underline{v_0}}, \boldvec{\underline{v_1}})\geq g^{(h)}_{j_0}(\boldvec{\underline{u'}}, \boldvec{\underline{v'}}) \geq \frac{1}{3} \Pi(P^{(h)}_{opt})$, a contradiction. As a consequence, $g^{(h)}_{j_0}(\boldvec{\underline{u'}}, \boldvec{\underline{v'}})$ is less than $\frac{1}{3} \Pi(P^{(h)}_{opt})$.

It then follows that $\Pi(P^{(h)}_{opt})\lceil_n - \Pi(P^{(h)}_{opt})\lceil_{j_0} \geq \frac{1}{3} \Pi(P^{(h)}_{opt})$ since, otherwise, $\Pi(P^{(h)}_{opt}) = \Pi(P^{(h)}_{opt})\lceil_n \leq \Pi(P^{(h)}_{opt})\lceil_{j_0-1} + g^{(h)}_{j_0}(\boldvec{\underline{u'}}, \boldvec{\underline{v'}}) + (\Pi(P^{(h)}_{opt})\lceil_{n} - \Pi(P^{(h)}_{opt})\lceil_{j_0}) < 3\cdot \frac{1}{3} \Pi(P^{(h)}_{opt}) = \Pi(P^{(h)}_{opt})$, yielding a contradiction.
Since $\Pi(P^{(h)}_{opt})\lceil_j \geq \Pi(P^{(h)})\lceil_j$ for all $j\geq j_0$,  if $P^{(h)}_{opt}(j,\boldvec{\underline{u'}})$ and $P^{(h)}_{opt}(j+1,\boldvec{\underline{v'}})$ are true, then the condition (2) of Definition \ref{def-MAXtauSNL} ensures that  $A^{(h)}_{j+1}(\boldvec{\underline{u'}}) \subseteq A^{(h)}_{j+1}(\boldvec{\underline{v_j}})$, where $A^{(h)}_{a}(\boldvec{\underline{w}}) = \{g_a^{(h)}(\boldvec{\underline{w}},\boldvec{\underline{w'}}\mid \boldvec{\underline{w'}}\}$.
Notice that $g^{(h)}_{j+1}(\boldvec{\underline{u'}}, \boldvec{\underline{v'}}) \in A^{(h)}_{j+1}(\boldvec{\underline{u'}})$ and
$g^{(h)}_{j+1}(\boldvec{\underline{v_j}}, \boldvec{\underline{v_{j+1}}}) \in A^{(h)}_{j+1}(\boldvec{\underline{v_j}})$.
Since the algorithm $\BB$ chooses $\boldvec{v_{j+1}}$ so that $g^{(h)}_{j+1}(\boldvec{\underline{v_j}}, \boldvec{\underline{v_{j+1}}})$ takes the maximum value in $A^{(h)}_{j+1}(\boldvec{\underline{v_j}})$, we conclude that $g^{(h)}_{j_0}(\boldvec{\underline{u'}}, \boldvec{\underline{v'}}) \leq g^{(h)}_{j_0}(\boldvec{\underline{v_{j_0-1}}}, \boldvec{\underline{v_{j_0}}})$  for all $j\geq j_0$.

For each number $j\geq j_0$, we write $t_{j}$ for the value $\Pi(P^{(h)}_{opt})\lceil_{j}- \Pi(P^{(h)})\lceil_{j}$. We then calculate the value $t_{j+1}-t_j$ as $t_{j+1}-t_j = (\Pi(P^{(h)}_{opt})\lceil_{j+1} - \Pi(P^{(h)}_{opt})\lceil_{j}) - (\Pi(P^{(h)})\lceil_{j+1} - \Pi(P^{(h)})\lceil_j) =  g^{(h)}_{j+1}(\boldvec{\underline{u'}}, \boldvec{\underline{v'}}) - g^{(h)}_{j+1}(\boldvec{\underline{v_j}}, \boldvec{\underline{v_{j+1}}}) \leq 0$, provided that $P^{(h)}(j,\boldvec{u'})$ and $P^{(h)}(j+1,\boldvec{v'})$ are true.
This implies that $t_{j+1} \leq t_j$ for all $j\geq j_0$. It then follows that $\sum_{j=j_0-1}^{n-1}t_{j+1} \leq \sum_{j=j_0-1}^{n-1}t_j$, which implies
$\Pi(P^{(h)}_{opt})\lceil_n - \Pi(P^{(h)}_{opt})\lceil_{j_0} \leq  \Pi(P^{(h)})\lceil_{n-1} - \Pi(P^{(h)})\lceil_{j_0-1}$.
We thus conclude that
$\Pi(P^{(h)})\lceil_{n-1} - \Pi(P^{(h)})\lceil_{j_0-1} \geq \Pi(P^{(h)}_{opt})\lceil_n - \Pi(P^{(h)}_{opt})\lceil_{j_0}$ for all $j\geq j_0$.

Since $\Pi(P^{(h)})\lceil_n\geq \Pi(P^{(h)})\lceil_{n-1}$, we conclude that $\Pi(P^{(h)})\lceil_n - \Pi(P^{(h)})\lceil_{j_0-1} \geq \Pi(P^{(h)}_{opt})\lceil_n - \Pi(P^{(h)}_{opt})\lceil_{j_0} \geq  \frac{1}{3} \Pi(P^{(h)}_{opt})$. It then follows that $\Pi(P^{(h)}) \geq \frac{1}{3} \Pi(P^{(h)}_{opt}) + \Pi(P^{(h)})\lceil_{j_0-1} \geq \frac{1}{3} \Pi(P^{(h)}_{opt})$. This is a clear  contradiction with our assumption that $\Pi(P^{(h)})< \frac{1}{3} \Pi(P^{(h)}_{opt})$.

Therefore, the algorithm should approximate $P_{opt}$ with approximation ratio of at most $3$.
\end{yproof}

%%%%%
%%%%%
\section{Brief Conclusion and Open Questions}\label{sec:open-problem}

Turing machines and circuit families have been used in the mainstream of computational complexity theory as basic computational models to solve various computational problems. In sharp contrast, a logical approach has taken to measure the complexity of these problems using the expressibility of specific logical sentences.

For a better understanding of $\nl$, the first logical approach was taken in \cite{Yam17a} using the notion of \emph{Syntactic NL} (or succinctly,  \emph{SNL}). $\snl$ sentences are characterized in the form of second-order ``functional''  existential quantifiers followed by first-order universal quantifiers together with two particular requirements called the \emph{second-order variable requirements}. Those sentences syntactically express certain types of languages.
The complexity class $\snl$ (and its variant $\snl_{\omega}$), consisting of all languages expressed syntactically by $\snl$ (and $\snl_{\omega}$) sentences
have a direct association with a practical, working hypothesis, known as the \emph{linear space hypothesis} (LSH).

In this work, we have continued the study of the structural properties of $\snl$. In particular, we have focused on three major issues:
(1) the expressibility of complementary problems of $\snl$ problems with an introduction of $\mu\snl$, which is a variant of $\snl$,
(2) the computational complexity of the monotone variant of $\snl$ (called $\mathrm{MonoSNL}$) together with its restriction (called $\mathrm{MonoBSNL}$), and
(3) the computational complexity of the optimization version of $\snl$ (called $\mathrm{MAXSNL}$) together with its variant (called $\mathrm{MAX}\tau\mathrm{SNL}$).

For the interested reader, we wish to raise a few important open questions associated with this work.
\renewcommand{\labelitemi}{$\circ$}
\begin{enumerate}\vs{-1}
  \setlength{\topsep}{-2mm}%
  \setlength{\itemsep}{1mm}%
  \setlength{\parskip}{0cm}%

\item Concerning various complexity classes discussed in this work, still unknown is any of the following class equalities: $\snl_{\omega} = \snl$, $\mathrm{MonoSNL}=\snl$, $\mathrm{MonoSNL}=\mathrm{CSP}_2$, $\mathrm{MonoBSNL}=\mathrm{BCSP}_2$,  $\mathrm{MAXSNL}=\mathrm{MAXSNP}$, and $\mathrm{MAX}\tau\mathrm{SNL} = \mathrm{MAXSNL}$. Proving the equalities or the inequalities will significantly deepen our understanding of SNL.

\item We have shown in Section \ref{sec:mu-operator} that the complementary problem of a particular $\snl$ problem is in $\mu\snl$. We thus wonder if the complementary problems of all problems in $\snl$ are in $\mu\snl$ (or more strongly, in $\snl$).

\item Given an $\snl$ sentence $\Phi$, we specifically call it an \emph{SNL sentence without comparison symbol} if the equality ($=$) as well as the less-than-or-equal symbol ($\leq$) is not used in $\Phi$.
    What is the computational complexity of decision problems expressed by such restrictive sentences?

\item In Definition \ref{def-mu-term}, we disallow any nested application of the $\mu$-operator for constructing $\mu$-terms. When we allow such a nested application of the $\mu$-operator for $\mu$-terms, then how dose the computational complexity of $\mu\snl$ change?

\item We have stated in Corollary \ref{BCSP-dichotomy} that the dichotomy theorem holds for binary CSPs. In the polynomial-time setting, MMSNP is known to satisfy the dichotomy theorem. Does $\mathrm{MonoBSNL}$ also enjoy the same dichotomy theorem?

\item We also expect a further study on $\mathrm{MAX}\tau\mathrm{SNL}$ and other natural subclasses of $\mathrm{MAXSNL}$ that are nicely contained in $\mathrm{APXL}$.
\end{enumerate}

%%%%%%
%%%%%%
%%%%%%%%%%%%%%%%%%%%%%%%%
%%%%%%%%%%%%%%%%%%%%%%%%%
\let\oldbibliography\thebibliography
\renewcommand{\thebibliography}[1]{%
  \oldbibliography{#1}%
  \setlength{\itemsep}{-2pt}%
}
\bibliographystyle{plainurl}
%\bibliographystyle{plain}

%%%%%%%%%%%%%%%%%
%%%%%%%%%%%%%%%%%
%%%%%%%%%%%%%%%%%%
%%%%%%%%%%%%%%%%%%
%%%%%%%%%%%%%%%%%%%%%%%%%%%%%%%%%%%%%%%%%%%%%%%%%%
%%%%%%%%%%%%%%%%%%%%%%%%%%%%%%%%%%%%%%%%%%%%%%%%%%%
\end{document}